\documentclass[a4paper,11pt,
]{article}

\usepackage{jheppub} 

\usepackage[T1]{fontenc} 
\usepackage{amsmath,amsthm}
\usepackage{amssymb,amsfonts}
\usepackage{graphicx}
\usepackage{dsfont}
\usepackage{makecell}
\usepackage{footnote}
\makesavenoteenv{tabular}
\makesavenoteenv{table}
\usepackage{enumitem}
\usepackage{relsize}
\usepackage{stmaryrd}
\usepackage{lmodern}
\usepackage{booktabs}
\usepackage{slantsc}
\usepackage{scalefnt}
\usepackage{subfigure}
\usepackage{xspace}
\usepackage{boxedminipage}
\usepackage{ifpdf}
\usepackage{multirow}
\usepackage{xifthen}
\usepackage{tensor}
\usepackage{array}
\usepackage{tabu}
\usepackage{tikz}
\usepackage{color}
\usepackage{diagbox}
\usetikzlibrary{arrows,matrix,calc,scopes,decorations.markings}
\allowdisplaybreaks[1]
\usepackage[mathcal]{euscript}
\newtheorem{definition}{Definition}

\newtheorem{proposition}[definition]{Proposition}

\newcommand{\be}{\begin{equation}}
\newcommand{\ee}{\end{equation}}
\newcommand{\bse}{\begin{subequations}}
\newcommand{\ese}{\end{subequations}}
\newcommand{\ket}[1]{\left|{#1}\right\rangle}

\newcommand{\Z}{\mathbb{Z}}

\newcommand{\inv}[1]{{{#1}^{-1}}}

\newcommand{\ii}{\mathrm{i}}

\newcommand{\U}{\mathcal{U}}
\newcommand{\C}{\mathcal{C}}
\newcommand{\scC}{\mathcal{C}}
\newcommand{\scD}{\mathcal{D}}
\newcommand{\scS}{\mathcal{S}}

\newcommand{\bpm}{\begin{pmatrix}}
\newcommand{\epm}{\end{pmatrix}}
\newcommand{\bmm}{\begin{matrix}}
\newcommand{\emm}{\end{matrix}}

\newcommand{\rarr}{\rightarrow}

\newcommand{\mcirc}{\,\circ\,}

\newcommand{\ox}{\otimes}
\newcommand{\aox}{{\,\ox\!_A\,}}

\newcommand{\isoto}{\xrightarrow{\sim}}

\newcommand{\rep}{\mathrm{Rep}}
\newcommand{\Hom}{\mathrm{Hom}}
\newcommand{\End}{\mathrm{End}}
\newcommand{\id}{\mathrm{id}}
\newcommand{\ginv}[1]{\overline{#1}}

\newcommand{\dimc}{{\dim_\C}}
\newcommand{\PFdimc}{\mathrm{PFdim}_{\mathcal{C}}}
\newcommand{\dima}{{\dim_A}}

\newcommand{\obj}[2]{{#1\!\left|_{#2}\right.}}
\newcommand{\grad}{\mathrm{\mathbf{gr}}}
\newcommand{\Rep}{\mathrm{Rep}}

\newcommand{\Aut}{\mathrm{Aut}}
\newcommand{\unit}{\boldsymbol{1}}

\newcommand{\qdbar}[2]{#1\overline{#2}}
\newcommand{\qdij}[1]{i_{#1}\overline{j_{#1}}}
\newcommand{\qab}[2]{\mathrm{A}^{#1}\mathrm{B}^{#2}}
\newcommand{\qba}[2]{\mathrm{B}^{#1}\mathrm{A}^{#2}}
\newcommand{\qdA}{\mathrm{A}}
\newcommand{\qdB}{\mathrm{B}}

\newcommand{\kab}[2]{\ket{\mathrm{A}^{#1}\mathrm{B}^{#2}}}
\newcommand{\skab}[2]{|{\mathrm{A}^{#1}\mathrm{B}^{#2}}\rangle}

\newtabulinestyle{dash=off 2pt}

{\null\,\vcenter\bgroup\scriptsize
\arraycolsep=.13885em
\hbox\bgroup$\array{@{}#1@{}}}
{\endarray$\egroup\egroup\,\null}

\newcolumntype{C}{>{\centering\arraybackslash} m{1.5em} }

\makeatletter
\newcommand*{\Relbarfill@}{\arrowfill@\Relbar\Relbar\Relbar}
\newcommand*{\xeq}[2][]{\ext@arrow 0055\Relbarfill@{#1}{#2}}
\makeatother

\input{mtikz.tex}

\title{Anyon Condensation: Coherent states, Symmetry Enriched Topological Phases, Goldstone Theorem, and Dynamical Rearrangement of Symmetry}
\date{\today}
\author[a,1]{Yuting Hu\note{Corresponding author}}
\author[a]{Zichang Huang}
\author[a,b]{Ling-yan Hung}
\author[a,b,1]{Yidun Wan}
\affiliation[a]{State Key Laboratory of Surface Physics,Department of Physics, Center for Field Theory and Particle Physics, and Institute for Nanoelectronic devices and Quantum computing, Fudan University, Shanghai 200433, China}
\affiliation[b]{Shanghai Qi Zhi Institute, Shanghai 200030, China}
\emailAdd{yuting.phys@gmail.com,hzc881126@hotmail.com, elektron.janethung@gmail.com, ydwan@fudan.edu.cn}

\abstract{
Although the mathematics of anyon condensation in topological phases has been studied intensively in recent years, a proof of its physical existence is tantamount to constructing an effective Hamiltonian theory. In this paper, we concretely establish the physical foundation of anyon condensation by building the effective Hamiltonian and the Hilbert space, in which we explicitly construct the vacuum of the condensed phase as the coherent states that are the eigenstates of the creation operators that create the condensate anyons. Along with this construction, which is analogous to Laughlin's construction of wavefunctions of fractional quantum hall states, we generalize the Goldstone theorem in the usual spontaneous symmetry breaking paradigm to the case of anyon condensation. We then prove that the condensed phase is a symmetry enriched (protected) topological phase by directly constructing the corresponding symmetry transformations, which can be considered as a generalization of the Bogoliubov transformation. 
}

\begin{document}

\maketitle
\flushbottom

\section{Introduction}

Anyon condensation has a ubiquitous presence in the study of topological order, particularly in 2+1 dimensions . It characterizes topological boundary conditions and defects between phases \cite{Bais2009a,Kitaev2012,Barkeshli,Barkeshli2013b,HungWan2014,HungWan2015a,Cong2017a,Wan2017}. It also provides mathematical maps (functors) between phases, giving new insights to the mathematical structures of, and relations between topological orders. Notably, the parent phase---the phase before anyon condensation---and its child phase---the condensed phase---also encapsulate data of a symmetry enriched topological (SET) \cite{maciejko2010,swingle2011,levinstern2012,Mesaros2011,Hung2013,Barkeshli2014c,Gu2014a}  or a symmetry protected topological (SPT) phase  \cite{affleck1987,Chen2011f}. The mathematical framework describing anyon condensation has been provided in \cite{Kitaev2012,Kong2013}. Yet the concept is often regarded as being mysterious particularly to physicists, since the mathematical description based on Frobenius algebra $A$ (describing the condensate) in a modular tensor category $\mathcal{C}$ and its representation (describing excitations in the condensate) do not appear to be directly related to states in a Hilbert space---a fundamental physical ingredient that would convince a physicist that the phenomenon exists. An explicit construction in the current context is comparable to the construction of the Laughlin wave-function that illustrates the existence of the fractional quantum hall state \cite{Laughlin1983}. The wave-function also lays out the possibility, if not a concrete path, towards realizing or simulating it in an experiment. 

Moreover, the notion of `condensation' would be immediately associated with spontaneous symmetry breaking in a physicist's mind. The analogy has been invoked in many discussions of anyon condensation. Nevertheless, if the analogy is more than just an analogy, it would be necessary to understand in what sense the new ground state is a coherent state of the condensate, and that one would expect that the Goldstone theorem to have a counterpart in these topological condensates.

In this paper, we would like to address these questions.  Along the way, our techniques will also find applications in understanding the symmetry enriched/protected topological (SET/SPT) phases through the lens of anyon condensation, allowing us to systematically construct symmetry action on explicit states in the Hilbert space of the child phase, completing the program instigated in \cite{HungWan2013a,Hung2013,Gu2014a}.

The key results of this paper are:
\begin{enumerate}
	\item  We build the explicit Hilbert space and construct creation operators that create the condensate anyons. Then we explicitly construct the vacua of the condensed phase as  the coherent states that are eigenvectors of the creation operators that create condensate anyons.
	\item We classify the inequivalent ground states that emerge for a given condensate in a twisted quantum double model (Dijkgraaf-Witten discrete gauge theory) based on the Tannaka duality. This is a generalization of the Goldstone theorem in the context of anyon condensation in topological orders, where the condensate could be magnetic/dyonic excitations with respect to the gauge group, unlike in the usual context in which the Goldstone theorem or the Higgs mechanism applies, the condensate are purely electric excitations. 
	\item We prove by construction that the condensed phase is a symmetry enriched (protected) topological phase, where the symmetry is a hidden global symmetry, described by a finite group $G$, and construct the corresponding symmetry transformations. The group $G$ is constructed from $A$ using Tannaka reconstruction theorem. Hence, there are $|G|$ sectors respectively labeled by the $G$ elements. The condensed phase lives in one sector---the physical sector. 
	\item The global $G$ transformation is decomposed into two actions: one mixes the anyons of the same energy within one sector and is characterized by a $G$-action on $\Rep A$, while the other maps one sector to another. This is a generalization of the Bogoliubov transformation for coherent states in the context of anyon condensation. 
	
\end{enumerate}

We shall briefly explain these key result in Section \ref{sec:exposition} before we derive them in details. We supply explicit examples to the results above. In the case of charge condensation and the case of dyon condensation in twisted quantum double models of arbitrary gauge group $G$ we supply a complete solution to the above construction.

If certain anyons in $\U$ are in a more general representations $\rho'$, they are organized into doublets or multiplets of $G$, that is, these anyons can be exchanged or permuted. When a physical boundary exists, an SET phase has stable gapless edge modes protected by the global symmetry.

This paper is organized as follows. Section \ref{sec:exposition} briefly explains our key results by exposing the relevant physics rather than mathematical details and presents a dictionary between the concepts in superconductor and their counterparts in anyon condensation. Section \ref{sec:toricCode} further aids the brief exposition by two simple examples, namely the $\Z_2$ toric code and doubled semion models. Section \ref{sec:coherentStates} explicitly constructs the coherent states that describe the condensates of anyons in the most generic setting. Section \ref{sec:symmTrans} constructs the symmetry transformation in condensed phases. With the doubled Ising topological phase, Section \ref{sec:exampleIsing} exemplifies that anyon condensation leads to a child phase being an SET phase. Section \ref{sec:chargeCond} (\ref{sec:dyonCond}) solves in the TQD model the problem of charge (dyon) condensation that leads to a group symmetry. Section \ref{sec:gauging} discusses the reverse process of anyon condensation, i.e., gauging an SET/SPT phase to a topological phase without a global symmetry. The appendices collect certain calculations too detailed to be part of the main text.


\section{A brief exposition of the main results}\label{sec:exposition}

Suppose the parent phase is characterized by a modular tensor category $\mathcal{C}$. The condensate in the condensed phase is characterized by a strongly separable commutative Frobenius algebra $A$ in $\mathcal{C}$. The anyons in the condensed phases are characterized by $\Rep^0 A$, where $\Rep^0 A$ is a full subcategory of $\Rep A$---the category of modules over $A$---generated by objects commuting with $A$. The rest objects not commuting with $A$ identify confined anyons.

To construct the explicit ground states of the condensed phase, we first construct an effective Hamiltonian to describe the dynamics of anyon condensation. We will start with an effective Hilbert space with basis states
\begin{equation}
	\label{eq:basisAJMU}
	\ket{a;j_1,j_2,\dots,j_N;\mu}.
\end{equation}
Each basis state consists of two parts: (1) $N$ anyons labeled by simple objects $j_n$ in $\C$ and (2) a \textit{condensate} anyon $a$ labeled by a simple direct summand\footnote{For example, say, if $A=a\oplus b\oplus c$, $a$, $b$, and $c$ are called direct summands in $A$, but a direct summand like $a\oplus b$ in $A$ is not simple.} in $A$. Such $N$ ``bare''-anyon states are in general degenerate, and $\mu$ labels the degenerate states. 

We proceed to construct the creation/annihilation operators of condensate anyons to build up the Hamiltonian. Define a \textit{creation operator} $W_n(a)$ as what extracts a condensate anyon $a$ from the condensate and fuse it to the $n$-th anyon $j_n$. See Fig. \ref{fig:reservoirPicture}. Define the effective Hamiltonian as
\begin{equation}\label{eq:HamiltonianPn}
	H_0=-\sum_n Q_n - \lambda\sum_n W_n,
	\quad
	Q_n\ket{j_n}=\delta_{j_n\in A}\ket{j_n},
\end{equation}
where $\lambda>0$, and $Q_n$ is a projector acting on the $n$-th anyon, with $\delta_{j_n\in A}=1$ if $j_n$ is a direct summand in $A$ and $0$ otherwise. The projectors $Q_n$ ensure that condensate anyons contribute no energy. Each $W_n$ is a projector as a linear combination of the creation operators $W_n(a)$ defined in Fig. \ref{fig:reservoirPicture}. The $W_n=1$ subspace is the observable Hilbert space, denoted by $\mathcal{H}^{W=1}$, of the condensed phase.

The observable topological charges are now characterized by $\Rep A$. The ground state(s) are superpositions of the states with arbitrary many condensate anyons. Intuitively we think all $a$-anyons are condensed in these states. These facts justify the terminology \textit{anyon condensation}.

\begin{figure}[!ht]
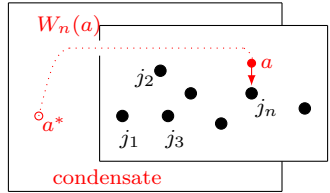

	\begin{center}
		\reserviorPicture
	\end{center}
	\caption{There are $N$ anyons labeled by the simple objects $j_n$ in $\mathcal{C}$. There is also a vacuum---a reservoir that maintains an indefinite number of condensate anyons (simple direct summands $a$ of $A$). The creation operator $W_n(a)$ creates a dual of $a$-anyon (can be viewed as a hole of $a$-anyon) in the reservoir and inserts an $a$ anyon to $j_n$.}
	\label{fig:reservoirPicture}
\end{figure}

The basis \eqref{eq:basisAJMU} can be transformed to another basis defined by the simultaneous eigenvectors of all $W_n(a)$, called the \textit{coherent states}. They should have eigenvalues $W_n=\pm 1$. The eigenvalues are called the \textit{coherent parameters}. Nevertheless, we only use the subspace $\mathcal{H}^{W=1}$ as the physical Hilbert space. The coherent states turn out to be a more natural basis in the condensed phase since they are labeled by the observables characterized by $\Rep A$.

Now we shall recover all the distinct sectors of ground states, obtaining a generalized version of the Goldstone theorem. 
We shall also show that there is a hidden global symmetry in the condensed phase and construct the corresponding symmetry transformations. This symmetry renders the condensed phase an SET or SPT phase.

Suppose $\mathcal{C}=\Rep(Q)$ is the representation category of some quantum group\footnote{More precisely, we mean quasi Hopf algebras or more generally weak Hopf algebras.} $Q$. Then the Hamiltonian $H_0$ is invariant under $Q$ transformations. In the condensed phase the ground state(s) do not necessarily manifest the $Q$-symmetry. This phenomenon is called \textit{spontaneous symmetry breaking}, or in other words, the $Q$-symmetry is hidden from the physical Hilbert space.

To see this, we add to $H_0$ an infinitesimal Hamiltonian $H_\epsilon$:
\begin{equation}\label{eq:infintyHamiltonian}
	H=H_0+H_\epsilon,
\end{equation}
where $H_\epsilon$ violates the $Q$-invariance. This $H_\epsilon$ prefers particular ground states that do not manifest the $Q$-symmetry. In the end of the day, we should take $\epsilon\to 0$.

To make precise the above mechanism, we focus on a special case where the spontaneously broken symmetry is a finite group $G$. This happens if $A$ is \textbf{absorbing}, i.e., $A\otimes A=A^{\oplus n}$ for $n\in \mathbb{Z}_+$ (precisely, it is the identity of $A$-modules in $\Rep A$). By Tannaka duality, we can reconstruct a finite group $G$ such that $A$ is equivalent to the regular representation of $G$. This derives a global transformations $U_g$, $g\in G$ on the condensate. We furthermore extend\footnote{Mathematically, this extension is equivalent to construct a fibre functor $\mathcal{C}\to Vec$ and a representation $U_g$ of $G$.} $U_g$ to act on all anyons in the condensed phase. We show that
\begin{enumerate}
	\item[(i)] $[H_0,U_g]=0$.
	\item[(ii)] The Hilbert space are divided into $|G|$ sectors. This collection of states can be understood as a discrete zero mode of a generalized Goldstone mode.
	The chosen $H_\epsilon$ does not commute with $U_g$ and hence selects a particular sector as the physical one.
\end{enumerate}
With this global $G$-symmetry, the resulting condensed phase is a symmetry enriched (protected) topological phase.

We emphasize the distinction between `symmetry' and `invariance'. The original $Q$-invariance of the Hamiltonian is never lost. The above construction implies that a $G$-sub-symmetry is hidden in the condensed phase. The survived symmetry is encoded in the topological observables characterized by $\Rep A$.

We call such a change of the manifestation of symmetry a \textit{symmetry rearrangement}. Concretely, the $Q$-invariance is rearranged as
\begin{enumerate}
	\item[(i)] The surviving symmetry encoded by topological observables characterized by $\Rep A$.
	\item[(ii)] The hidden global symmetry transformation $U_g$ is decomposed into two parts(See Fig. \ref{fig:symmetryArrangement}):
	\begin{enumerate}
		\item The mixing of anyons with the same energy within one sector.
		\item When expressed in terms of coherent states, each sector is labeled by coherent parameters $\{g_n\}$. The $U_g$ maps one sector to another by  $\{g_n\}\mapsto \{gg_n\}$.
	\end{enumerate}
\end{enumerate}

\begin{figure}[!ht]
	\centering\includegraphics[width=0.8\textwidth]{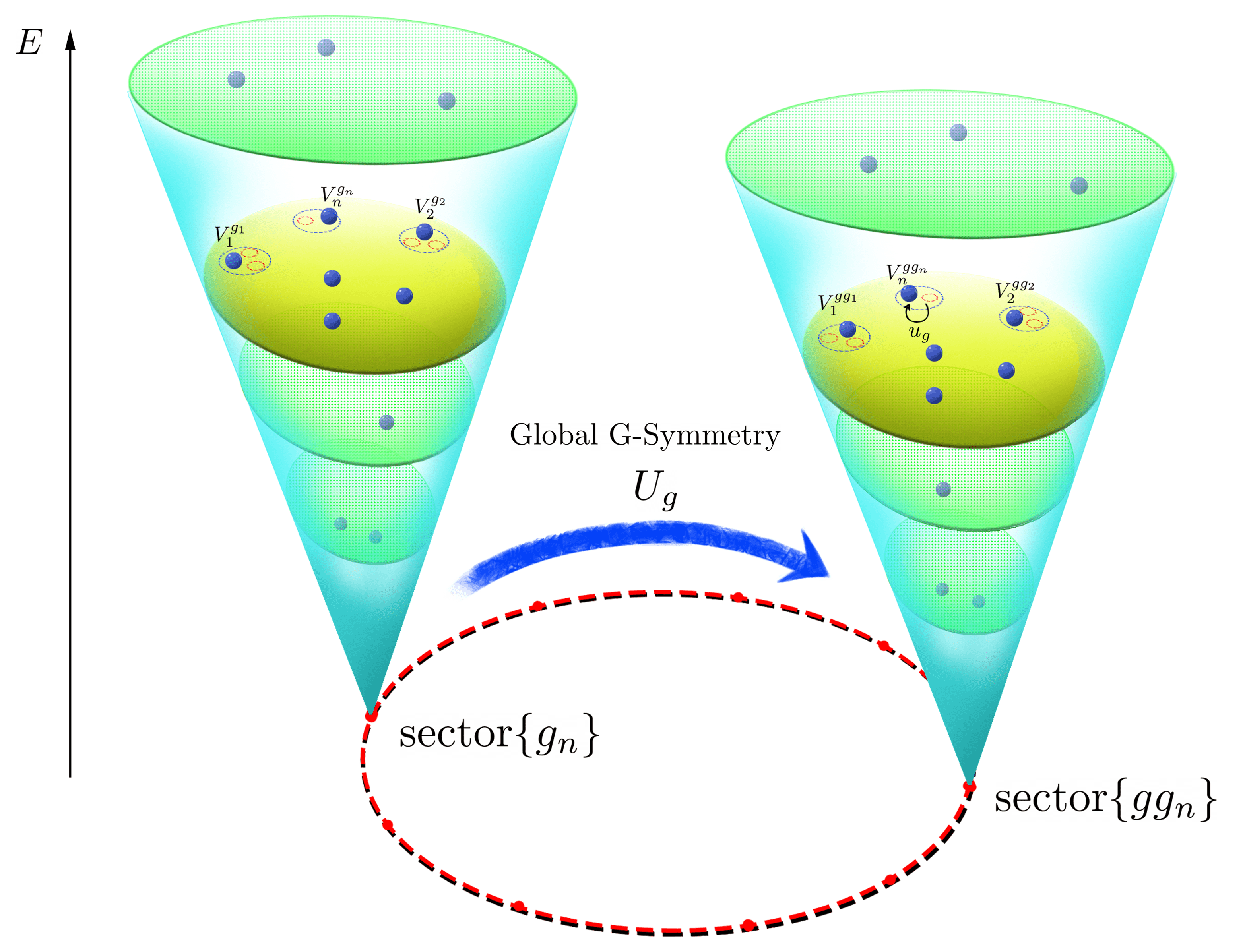}
	\caption{A global transformation $U_g$ maps the sector with parameter $\{g_n\}$ to one with $\{gg_n\}$. An anyon with topological charge $V^{g_n}_n$ is transformed by $V^{g_n}_n\to T_g(V^{g_n}_n)$, where $T$ is a group action of $G$ on $\Rep A$.}
	\label{fig:symmetryArrangement}
\end{figure}

Mathematically, the decomposition of the global transformation $U_g$ is characterized by the \textit{equivariantization category} $(\Rep A)^G$ whose objects are $(V,u_g)$ for $V\in \Rep A$ and $u_g:T_g(V)\to V$ is an equivariant structure. The mapping from one sector to another under $U_g$ can be formulated by $T_g$ where $T$ is a group action on $\Rep A$. The mixing of anyons with the same energy within each sector can be formulated by the equivariant structure $u_g:T_g(V)\to V$. The equivalence $(\Rep A)^G\cong \mathcal{C}$ implies that the $Q$-invariance can be recovered from the hidden global $G$-symmetry.

As concrete examples of the above mechanism, we explicitly construct the effective theory of the charge condensation in the twisted quantum double models, leading to SPT phases.

\begin{table}[!h]\small
	\begin{center}
		\begin{tabular}{l|cc}
			\toprule 
			& anyon condensation concepts & example: superconductor\\
			\midrule
			original invariance & modular tensor category $\mathcal{C}$ & $\mathcal{C}=\Rep_{D(U(1))}$\\
			\hline
			condensate (coherent state) &  $A$ anyons &  cooper pairs\\
			\hline
			hidden global symmetry & \makecell[c]{finite group $G=\Aut(A)$ \\ reconstructed from $A$} & $U(1)$\\
			\hline
			symmetry rearrangement & $\mathcal{C}\to (\Rep A)^G$ & \makecell[c]{$U(1)$ gauge invariance \\ $\Rep A=Vec_{U(1)}$, $G=U(1)$    }\\
			\hline
			confined sectors & $\Rep^g A$, $g\neq 1$  & $U(1)$-flux \\
			\hline
			\makecell[l]{coherent parameters\\ (not observable)} & $\{g_n\}$  & fixed gauge\\
			\hline
			order parameter & $g_0$ & $U(1)$ phase   \\
			\bottomrule
		\end{tabular}
	\end{center}
	\caption{Superconductor as an example of anyon condensation. The superconductor has a generalized setup, where $\mathcal{C}$ is no longer a finite category, and quantum fields have infinitely many degrees of freedom. The gauge $U(1)$ invariance is thus rearranged to a field translation of the phason field, which seems different from our results (the $G$ invariance becomes the hidden symmetry $G$). }\label{tab:superconductivity}
\end{table}

\section{Condensed states: in toric code model and double semion model}\label{sec:toricCode}

In this section we study the bosonic anyon condensation in the toric code model and double semion model as two simplest examples to demonstrate our results. We will construct the coherent states and show the dynamical symmetry rearrangement.

\subsection{Example: in toric Code Model}

In the toric code model, elementary excitations carry topological charges: vacuum $\mathbf{1}$, $\mathbb{Z}_{2}$ charge $e$, $\mathbb{Z}_{2}$ flux $m$, and the charge-flux composite $em$. Let us condense $e$. After the condensation, two topological charges remains: $A=\mathbf{1} \oplus e$ and $B=m\oplus em$. 



In what follows we consider charge condensation, so $A=1\oplus e$, where $e$ is the charge excitation in the toric code model. The new vacuum, i.e., the vacuum of the condensed phase, will be the superposition of states of arbitrary many $e$'s.

We start with the basis vectors
\begin{equation}
	\label{eq:basisJ}
	|a;j_{1} ,j_{2} ,\dotsc ,j_{N} \rangle ,
\end{equation}
where $a\in\{1,e\}$, and $j_{n}\in \{1,e,m,em\}$ the topological charge of the $n$-th anyon (Fig. \ref{fig:reservoirPicture}).  Note that the degeneracy label $\mu$ in Eq. \eqref{eq:basisAJMU} does not exist here.

By using the basis \eqref{eq:basisJ}, we assume the total topological charge could be either $1$ or $e$. We may have other types of constraints on the total topological charge depending on the boundary conditions of the system. Nevertheless, such constraints will not affect our results derived in this paper. Thus, for simplicity, we do not specify boundary conditions in the rest of the paper.


\subsubsection{Creation Operator and coherent states}

In this example, define the \textbf{creation operator} $W_{n}(e)$ by
\begin{equation}
	W_{n}(e) |a;j_{1} ,j_{2} ,\dotsc , j_{n},\dots,j_{N} \rangle =|e\ox a;j_{1} ,j_{2} ,\dotsc , e\ox j_{n},\dots,j_{N} \rangle ,
\end{equation}
with $n=1,\dots,N-1$, which creates a pair of $e$ in the condensate and fuse one of them to $j_{n}$.

In the following, we will construct the simultaneous eigenvectors of $W_{n}(e)$. We will use them to construct the condensed states later.
Since $W^{2}_{n}(e) =1$, $W_{n}(e)$ has eigenvalues $\pm 1$. 

We first consider the case with positive eigenvalues: $W_{n}(e) |\Phi \rangle = |\Phi \rangle $ $\forall n$. The eigenspace $\mathcal{H}^{W=1}$, i.e., the physical Hilbert space, has basis vectors
\begin{equation}
	\label{eq:basisToricCodeAVV}
	|A;V_{1} ,V_{2} ,\dotsc ,V_{N} \rangle ,
\end{equation}
which abbreviates $\ket{A}\otimes\ket{V_1}\otimes\dots\otimes\ket{V_N}$, where $V_{n} =A,B$, with $\ket{A}=\frac{\ket{1}+\ket{e}}{\sqrt{2}}$, and $B=\frac{\ket{m}+\ket{em}}{\sqrt{2}}$. Note that we only allow the  basis states with even occurrences of $B$. For example, the state $|A;A,B,\dotsc ,B\rangle$ is expanded as
\begin{equation}
	|A;A,B,\dotsc ,B\rangle =2^{-(N+1)/2}\sum_{a=1,e}\sum _{j_{1} =1,e}\sum _{j_{2} =m,em} \dotsc \sum _{j_{N} =m,em}  |a;j_{1} ,j_{2} ,\dotsc ,j_{N} \rangle ,
\end{equation}
in terms of the basis vectors \eqref{eq:basisJ}. 

The basis vectors \eqref{eq:basisToricCodeAVV} span the physical Hilbert space of the condensed phase, where $\ket{A;A,A,\dots,A}$ is the vacuum that is the superposition of arbitrary many condensate anyons.

Now define a $\mathbb{Z}_{2}$ operator $u_-$ by
\begin{equation}
	\label{eq:Z2ActionToricCode}
	\begin{aligned}
		u_- \ket{1} & =\ket{1},\\
		u_- \ket{e} & =-\ket{e},\\
		u_- \ket{m} & =\ket{m},\\
		u_- \ket{em} & =-\ket{em}.
	\end{aligned}
\end{equation}
Let $u_+=\id$, such that $\{u_g,g=\pm\}\cong \mathbb{Z}_2=$ under multiplication. We also define
\begin{equation}\label{eq:Vnpm}
	\ket{V^g}:=u_g\ket{V},
\end{equation}
for $V=A,B$. All simultaneous eigenvectors of $W_{n}(e)$ are generated from $|A;V_{1} ,V_{2} ,\dotsc ,V_{N} \rangle $ by $u_g$:
\begin{equation}
	\label{eq:eigenEquationWn}
	\begin{aligned}
		& W_{n}(e) |A^{g_0}; V_{1}^{g_1} , V_{2}^{g_2} ,\dotsc ,V_{n}^{g_n} ,\dotsc ,V_{N}^{g_N} \rangle \\
		= \, & ( g_{0} g_{n})  |A^{g_0}; V_{1}^{g_1} , V_{2}^{g_2} ,\dotsc ,V_{n}^{g_n} ,\dotsc ,V_{N}^{g_N} \rangle
	\end{aligned}.
\end{equation}
The eigenvalues are computed from the commutator relation
\begin{equation}\label{eq:WFFCommutatorToricCode}
	W_nu_{g_0}u_{g_n} =( g_{0} g_{n}) u_{g_0}u_{g_n} W_n.
\end{equation}

We call these eigenvectors the \textbf{coherent states} of the creation operators $W_{n}(e)$, and call the eigenvalues $\{g_{0};g_{1} ,g_{2} ,\dotsc ,g_{N}\}$ the coherent parameters. 


We summarize the orthonormal condition of the coherent states:
\begin{align}\label{eq:orthonormalCoherentStates}
	&\langle A^{g_0}; V_{1}^{g_1} ,\dotsc ,V_{N}^{g_N} \ket{A^{g'_0}; {V'}_{1}^{g'_1} ,\dotsc,{V'}_{N}^{g'_N}}
	\nonumber\\
	=\, &\delta_{g_0,g'_0}\delta_{g_1,g'_1}\dots\delta_{g_N,g'_N}\delta_{V_1,V'_1}\dots \delta_{V_N,V'_N}
\end{align}

\subsubsection{Effective Hamiltonian and the global symmetry}

We define the effective Hamiltonian as
\begin{equation}\label{eq:effectiveHamiltonianToricCode}
	H_0=-\sum_n Q_n - \sum_n W_n,
	\quad
	Q_n\ket{j_n}=\delta_{j_n\in A}\ket{j_n},
\end{equation}
where $\delta_{j_n\in A}=1$ if $j_n\in \{1,e\}$ and $0$ otherwise.

We shall add to $H_0$ an extra infinitesimal Hamiltonian
\begin{equation}\label{eq:infiHamToricCode}
	H_\epsilon = \epsilon \sum_n\Phi_n,
	\quad
	\Phi_n\ket{A^{g_0};V^{g_1}_1,\dots,V^{g_N}_N}=\delta_{g_n,+}\ket{A^{g_0};V^{g_1}_1,\dots,V^{g_N}_N}.
\end{equation}
We should take the limit $\epsilon\to 0$ at the end of calculation.

We can construct the global symmetry transformation $U_g$ by
\begin{equation}\label{eq:globalTransform}
	U_g|A^{g_0}; V_{1}^{g_1} , \dotsc ,V_{N}^{g_N} \rangle
	= |A^{gg_0}; V_{1}^{gg_1} , \dotsc ,V_{N}^{gg_N} \rangle.
\end{equation}
Obviously, $[H_0,U_g]=0$. In the limit $\epsilon\to 0$, there are two ground states
\begin{equation}\label{eq:twoGroundStatesToricCode}
	\ket{A^\pm;A^\pm,\dots,A^\pm}.
\end{equation}
However, due to the presence of $H_\epsilon$, the actual ground state is
\begin{equation}\label{eq:GroundStateToricCode}
	\ket{A^+;A^+,\dots,A^+}.
\end{equation}

\subsection{Example: in double semion model}

The double semion model is identical to the twisted quantum double model defined by $G=\mathbb{Z}_2$ and a non-trivial 3-cocycle over $G$, we use the same symbols to label the topological charges as in the toric code model. However, there is a difference: these topological charges obey nontrivial braiding statistics: $e$ is called the doubled semion, $m$ the semion, and $em$ the anti-semion.

We condense the doubled semion $e$, using the Frobenius algebra $A=1\oplus e$. The creation operator $W_n$ is defined as in Eq. \eqref{eq:Z2ActionToricCode}. To construct the coherent states, we define the $\mathbb{Z}_2$ action by
\begin{equation}
	\label{eq:Z2ActionDoubledSemion}
	\begin{aligned}
		\begin{aligned}
			u_- \ket{1} & =\ket{1},\\
			u_- \ket{e} & =-\ket{e},\\
			u_- \ket{m} & =\ii\ket{m},\\
			u_- \ket{em} & =-\ii\ket{em}.
		\end{aligned}
	\end{aligned}
\end{equation}
and set $u_+=\mathrm{id}$.

The local $\mathbb{Z}_2$ action is different from that in the toric code model. When acting on $m$ and $em$, $u_-$ is a projective representation: $u_- u_- = -u_+$.

Define $A^{\pm}=\frac{\ket{1}+\ket{e}}{\sqrt{2}}$ and $B^+=\frac{\ket{m}+\ket{em}}{\sqrt{2}}$, $B^-=\ii\frac{\ket{m}-\ket{em}}{\sqrt{2}}$. Now the coherent states has the same form as in the toric code case:
\begin{align}\label{eq:coherentStateDoubledSemion}
	& W_{n}(e) |A^{g_0}; V_{h_1}^{g_1} , V_{h_2}^{g_2} ,\dotsc ,V_{h_n}^{g_n} ,\dotsc ,V_{h_N}^{g_N} \rangle \nonumber\\
	= \, & ( g_{0} g_{n}) |A^{g_0}; V_{h_1}^{g_1} , V_{h_2}^{g_2} ,\dotsc ,V_{h_n}^{g_n} ,\dotsc ,V_{h_N}^{g_N} \rangle,
\end{align}
where we set $V_{h=+}^g=A^g$ and $V_{h=-}^g=B^g$, for convenience of future computation.

The global symmetry transformation is given by
\begin{align}\label{eq:symmetryTransformDoubledSemion}
	&U_g|A^{g_0}; V_{h_1}^{g_1} , V_{h_2}^{g_2} ,\dotsc ,V_{h_n}^{g_n} ,\dotsc ,V_{h_N}^{g_N} \rangle
	\nonumber\\
	=\, & \mu_g(h_1,h_2)\mu_g(h_1h_2,h_3)\dots \mu_g(h_1h_2\dots h_{N-1},h_N)
	\nonumber\\
	&\beta_{h_1}(g,g_1)\beta_{h_2}(g,g_2)\dots\beta_{h_N}(g,g_N)
	\nonumber\\
	&\qquad\quad
	|A^{gg_0}; V_{h_1}^{gg_1} , V_{h_2}^{gg_2} ,\dotsc ,V_{h_n}^{gg_n} ,\dotsc ,V_{h_N}^{gg_N} \rangle,
\end{align}
where $\beta_{h_n}(g,g_n)=-1$ if $h_n=g=g_n=-$ and 1 otherwise, $\mu_g(h_1,h_2)=-1$ if $g=h_1=h_2=-$ and $1$ otherwise.

We shown the charge condensation in the twisted quantum double models with input data $\mathbb{Z}_2$. In Section \ref{sec:chargeCondensationTQD}, we will deal with a generic twisted quantum double model.

\section{Effective theory of anyon condensation in generic cases}
In this section, we will construct an effective theory of anyon condensation in generic cases. We will construct coherent states for the condensed phases in generic cases, and construct the hidden global symmetry transformation.

\subsection{Coherent States}\label{sec:coherentStates}

We take a topological phase characterized by a modular tensor category $\mathcal{C}$ and study an anyon condensation characterized by a strongly separable commutative Frobenius algebra $A\in \mathcal{C}$. Denote by $L_\mathcal{C}$ the set of all representative simple objects of $\mathcal{C}$, by $L_A\subset L_{\mathcal{C}}$ the subset of those that are the direct summands in $A$, and by $L_{\Rep A}$ the set of all representative simple objects of $\Rep A$.

We will start with an effective Hilbert space with basis states presented by
\begin{equation}\label{eq:quasiparticleStates}
	\quasiparticleStates.
\end{equation}
They are the basis vectors of $\mathrm{Hom}_{\mathcal{C}}( a_0,a_1\otimes j_{1} \otimes j_{2} \otimes \dotsc \otimes j_{N})$, where $a_0,a_1,a_2\in L_{A}$, $j_1,\dots,j_N,i_1,\dots\in L_{\mathcal{C}}$ and $L_{\mathcal{C}}$. (For computing with concrete tensors, we fix an appropriate basis of $\mathrm{Hom}_{\mathcal{C}}(i,j\otimes k)$ for all $i,j,k\in L_{\mathcal{C}}$. In the diagram, we have also suppressed the multiplicity degrees of freedom at trivalent vertices.) The above diagram depicts the relevant quantum observables: a topological charge $j_n$ with the $n$-th quasiparticle, and a total topological charge $i_{n-1}$ of the subsystem of first $n$ quasiparticles. The total topological charge is $a_0$. In a real system, the constraint on the total topological charge depends on the boundary condition. We shall not dwell on such constraints since they will not affect our results.

We consider a special situation, where the $a$-anyons only appears in the ground states. Mathematically, if $A$ is \textbf{absorbing}, i.e., $A\otimes A=A^{\oplus n}$ for $n\in \mathbb{Z}_+$ as $A$-modules in $\Rep A$, then a finite group $G$ can be constructed \cite{Deligne1991,Muger2004} according to Tannaka reconstruction theorem. Namely, there exits a unique finite group $G$, such that $A$ is identified with the regular representation of $G$. 
Rigorously, we define $G$ by
\begin{equation}\label{eq:AutAGroupCopy}
	G=\Aut(A)=\left\{g\in \mathrm{End}_{\mathcal{C}}(A) \left|\quad \autAcond\quad\right.\right\},
\end{equation}
where the trivalent vertex presents the multiplication $A\ox A\to A$ and the small circle the unit of $A$. See Appendix \ref{sec:reconstruct} for more details. Throughout this paper, a red line (colored online) presents an $A$. The exact conditions on $A$ for the reconstruction theorem to be valid are given by Proposition \ref{prop:absorbingA} in Appendix \ref{sec:reconstruct}.

We will define the \textit{creation operator} $W_n$ using the algebra structure of the Frobenius algebra $A$. The operator $W_n$ creates a pair of $A$ anyons in the condensate and fuse one to the $n$-th quasiparticle. To make sense the fusion process, the $n$-th quasiparticle must be represented by an $A$-module $V_n\in \mathrm{Rep}A$. We thus define
\begin{equation}\label{eq:creationOperatorWnGenericCase}
	W_n=\frac{1}{\mathrm{dim}_{\mathcal{C}}A}\creationOperatorWn,
\end{equation}
where the red trivalent vertex on the left is the coproduct $\Delta:A\to A\otimes A$, followed by a multiplication $A\ox A\to A$ and the module action $A\ox V_{n}\to V_{n}$. In fact, $W_n$ is a projector on $A\otimes V_{n}$. 
In the example of the toric code model, this generic projector in Eq. \eqref{eq:creationOperatorWgnGenericCase} is identified with $(1+W_n(e))/2$.

We extend $W_n$ to a complete set of projectors respectively labeled by the group elements of $G$. Define
\begin{equation}\label{eq:creationOperatorWgnGenericCase}
	W^g_n=\frac{1}{\mathrm{dim}_{\mathcal{C}}A}\creationOperatorWnAA,
\end{equation}
for $g\in G$. Particularly, $W_n=W_n^1$. The elements in the set $\{W^g_n\}$ are orthonormal:
\begin{equation}\label{eq:orthogonalWn}
	W^g_nW^h_n=\delta_{g,h}W^g_n,
\end{equation}
\begin{equation}\label{eq:completeWn}
	\sum_{g\in G}W^g_n=\mathrm{id}.
\end{equation}

In the following We abbreviate $(V,\rho)\in L_{\Rep A}$ to $V$. The $W_n$-invariant states are spanned by the basis states depicted as
\begin{equation}\label{eq:quasiparticleStatesRepA}
	\quasiparticleStatesRepA,
\end{equation}
where $V_1,V_2,\dots,U_1,U_2,\dots\in L_{\mathrm{Rep}A}$. The $g\in G$ can be chosen arbitrarily. The diagram is a composition of a morphism in $\mathrm{Hom}_{A}( A,A\aox V_{1} \aox V_{2} \aox \dotsc \aox V_{N})$ and a $\mathcal{C}$-morphism $g\in \Aut(A)= G$. In the diagram, we suppress the indices due to the multiplicities in $\mathrm{Hom}_A$. The diagram depicts the relevant quantum observables: the topological charge $V_n$ at the $n$-th quasiparticle, and a total topological charge $U_{n-1}$ of the first $n$ quasiparticles.

One can verify  $[W^g_n,W^{g'}_{n'}]=0$ straightforwardly. Hence, we can define the \textit{coherent states} as the simultaneous eigenvectors of $W^g_n$ for all $n$ and $g$. It is convenient to use them as the new basis instead of that in Eq. \eqref{eq:quasiparticleStates}.

We now explicitly construct the coherent states, in a way similar to that in the previous subsection: We can generate the coherent states by acting some local operators $u_g$ on the $W_n$-invariant basis \eqref{eq:quasiparticleStatesRepA}.

A group element $g$ is an automorphism of $A$ and induces a group structure on any generic object in $\Rep A$. In Appendix \ref{sec:IsomorphismPsi}, we prove that for every $(V,\rho_V)\in L_{\Rep A}$ and $g\in G$, there exits an object $(V,\rho'_V)\in L_{\Rep A}$, and a unique $\mathcal{C}$-isomorphism $u_g:V\to V$, such that
\begin{equation}\label{eq:uniqueMorphismVVprimeABappCOPY}
	\uniqueMorphismVVprimeAB.
\end{equation}

The uniqueness of ${u_g}$ implies that ${u_g}u_h\propto u_{gh}$ on each $(V,\rho_V)$. Denote by $\beta_{(V,\rho_V)}(g,h)$ the coefficient of the proportionality:
\begin{equation}\label{eq:betaVgh}
	\betaVgh.
\end{equation}
Because the action of ${u_g}$ is associative, $\beta_{(V,\rho_V)}$ is a 2-cocycle.

The coherent states are generated by acting $u_g$ on $W_n$-invariant basis states. We define the basis by
\begin{equation}\label{eq:psigBasisRepA}
\ket{g_0A;g_1V_1,g_2V_2,\dots,g_NV_N;U_1,U_2,\dots;\alpha}:=\quasiparticleStatesRepAAH,
\end{equation}
where $V_1,\dots,V_N,U_1,\dots\in L_{\Rep A}$,
\begin{equation}\label{eq:alphaBasisCoherentStates}
	\begin{aligned}
		&\alpha_1\in \Hom_A(U_1,V'_1\aox V'_2),\\
		&\alpha_2\in \Hom_A(U_2,U_1\aox V'_3),\\
		&\dots
	\end{aligned}
\end{equation}
and $u_{\ginv{g_n}}:V'_n\to V_n$ with $\ginv{g_n}\equiv g_n^{-1}$.

The defining property \eqref{eq:uniqueMorphismVVprimeABappCOPY} implies an eigenvalue equation
\begin{equation}\label{eq:WngPhigV}
	\begin{aligned}
		&W_n^{(g_{n}\ginv{g_0})}\ket{g_0A;g_1V_1,g_2V_2,\dots,g_NV_N;U_1,U_2,\dots;\alpha}\\
		=&\ket{g_0A;g_1V_1,g_2V_2,\dots,g_NV_N;U_1,U_2,\dots;\alpha}.
	\end{aligned}
\end{equation}
Hence, we dub the states \eqref{eq:psigBasisRepA} the \textit{coherent states} with the coherent parameters $g_1,\dots,g_N$.

The coherent states form a new basis of the effective Hilbert space. Eq. \eqref{eq:psigBasisRepA} is an orthonormal basis of $\Hom_{\mathcal{C}}(A,A\ox V_1\ox V_2\ox \dots \ox V_N)$. 
%
To see this, we need only look at the decomposition of $\Hom_{\mathcal{C}}(U_1,V_1\ox V_2)$:
\begin{equation}\label{eq:proofCoherentStatesOrthonormalBasis}
	\begin{aligned}
		&\Hom_{\mathcal{C}}(U_1,V_1\ox V_2)\\
		=&\Hom_A(U_1,F(V_1\otimes V_2))\\
		=&\Hom_A(U_1,FV_1\aox FV_2)\\
		=&\bigoplus_{g_1,g_2\in G}\Hom_A(U_1,T_{g_1}V_1 \aox T_{g_2}V_2),
	\end{aligned}
\end{equation}
where $F:\mathcal{C}\to\Rep A$ is the free functor given in Eq. \eqref{eq:FXaction}, and $h_1$ can be any group element. The first equality above follows Eq. \eqref{eq:FXidempotent}, and in the third equality, use is made of the decomposition $FV=\oplus_g T_g V$ as in Proposition \ref{prop:FVdecomposition}(i) of Appendix \ref{sec:orbitsUnderGroupAction}.

Using the creation operator $W_n$, we can define the effective Hamiltonian as in Eq. \eqref{eq:HamiltonianPn}, and we obtained the a physical Hilbert space spanned by $\ket{gA;gV_1,gV_2,\dots,gV_N;U_1,U_2,\dots;\alpha}$, where $g$ is randomly chosen by the system.

\subsection{Symmetry transformations}\label{sec:symmTrans}

We will show that the global symmetry is hidden (spontaneously broken) in the condensed states using the coherent states constructed in the previous section, and construct the corresponding symmetry transformation.

To study how $u_g$ acts, we first introduce a grading. Namely, for every simple $(Z,\rho)\in\mathrm{Rep}A$, define the grading $\grad(Z,\rho)\in G$ by
\begin{equation}\label{eq:gradingZRho}
	\grad(Z,\rho)=\frac{1}{\dima Z\,\dimc A}\gradingXinRho.
\end{equation}
The grading is conjugated under $u_g$: $\grad u_g(Z,\rho)=g\grad(Z,\rho)\ginv{g}$. See Appendix \ref{sec:RepAGradedByG} for more properties. 

We shall use $u_g$ to study the transformations of the coherent states, so we need to know how to transform $\alpha\in\mathrm{Hom}_A( (U,\rho_U), (V,\rho_V)\otimes_A(W,\rho_W))$. Define ${^{g} \alpha}$ by
\begin{equation}\label{eq:MapUVW}
	\UVWalphaAD:=\UVWalphaAC,
\end{equation}
which is $\mu_{\ginv g}(U,V,;W)(\alpha)$ as defined in Eq. \eqref{eq:mugUVW} in Appendix \ref{sec:IsomorphismPsi}.

According to Eq. \eqref{eq:betaVgh}, we have
\begin{equation}\label{eq:psigPsihPsigh}
	{^g}({^{h}}\alpha)=\frac{\beta_U(\ginv g,\ginv h)}{\beta_V(\ginv g,\ginv h)\beta_W(\ginv g,\ginv h)}({^{gh}}\alpha).
\end{equation}

The symmetry transformation $U_g$ is defined by
\begin{equation}\label{eq:globalSymmetryTransform}
	U_g:\quasiparticleStatesRepAAH\mapsto\quasiparticleStatesRepAAI,
\end{equation}
where the RHS is expanded as
\begin{equation}\label{eq:expandTg}
	\quasiparticleStatesRepAAE = \prod_{n=1}^{N}\beta_{V_n}(\ginv{g_n},\ginv{g})\quasiparticleStatesRepAAF .
\end{equation}

\subsection{Example: Double Ising model}
\label{sec:exampleIsing}

The topological phase characterized by the quantum doubel of the Ising category, with simple objects $i\overline{j}\in L_{\mathcal{C}}$, where $i,j=0,1,2$ (usually denoted as $1,\sigma,\psi$ in the literature). We fix a basis of $\Hom_{\mathcal{C}}$ such that
\begin{equation}\label{eq:basisTransform}
	\basisHomC = \sum_{x_6\in L_{\mathcal{C}}}v_{x_3}v_{x_6}G^{x_1x_2x_3}_{x_4x_5x_6} \basisHomCAA,
\end{equation}
for $x_n=\qdij{n}\in L_{\mathcal{C}}$. Here $v_{\qdbar{i}{j}}=v_i v_j$,
\begin{equation}\label{eq:vIsing}
	v_0=1,v_{1}=2^{1/4},v_{2}=1,
\end{equation}
\begin{equation}\label{eq:6jQD}
	G^{\qdij{1}\qdij{2}\qdij{3}}_{\qdij{4}\qdij{5}\qdij{6}}
	=G^{i_1i_2i_3}_{i_4i_5i_6}\overline{G^{j_1j_2j_3}_{j_4j_5j_6}},
\end{equation}
and
\begin{equation}\label{eq:6jIsing}
	\begin{aligned}
		&G_{000}^{000}=1, \quad G_{111}^{000}=\frac{1}{\sqrt[4]{2}}, \quad G_{222}^{000}=1, \quad G_{011}^{011}=\frac{1}{\sqrt{2}} \\
		&G_{122}^{011}=\frac{1}{\sqrt[4]{2}}, \quad G_{211}^{011}=\frac{1}{\sqrt{2}}, \quad G_{022}^{022}=1, \quad G_{112}^{112}=-\frac{1}{\sqrt{2}}.
	\end{aligned}
\end{equation}

Let $A=\qdbar{0}{0}\oplus\qdbar{2}{2}$, with the multiplication $A\ox A\to A$ expressed in the above basis by
\begin{equation}\label{eq:multiplicationAIsing}
	f = 
	\multiplicationAIsing + \multiplicationAIsingAA
	+ \multiplicationAIsingAB + \multiplicationAIsingAC.
\end{equation}

There are six (inequivalent representative) simple objects in $\Rep A$. Using the above basis, they are expanded as
\begin{equation}\label{eq:simpleObjectsInRepA}
	(V,\rho_V=\sum_{ax_1x_2}\rho^a_{x_1x_2}\rhoInRepA),
\end{equation}
where $a$ runs over the summands of $A$ and $x_1,x_2$ over those of $V$. The six simple objects are listed as
\begin{equation}\label{eq:sixObjectsInRepAIsing}
	\begin{array}{|c|c|}
		\hline
		0 \bar{0}\oplus 2 \bar{2} & 
		\begin{array}{cccc}
			\rho _{0 \bar{0},0 \bar{0}}^{0 \bar{0}}=1 & \rho _{2 \bar{2},2 \bar{2}}^{0 \bar{0}}=1 & \rho _{0 \bar{0},2 \bar{2}}^{2 \bar{2}}=\frac{1}{\xi } & \rho _{2 \bar{2},0 \bar{0}}^{2 \bar{2}}=\xi  \\
		\end{array}
		\\
		\hline
		0 \bar{2}\oplus 2 \bar{0} & 
		\begin{array}{cccc}
			\rho _{0 \bar{2},0 \bar{2}}^{0 \bar{0}}=1 & \rho _{2 \bar{0},2 \bar{0}}^{0 \bar{0}}=1 & \rho _{0 \bar{2},2 \bar{0}}^{2 \bar{2}}=\frac{1}{\xi } & \rho _{2 \bar{0},0 \bar{2}}^{2 \bar{2}}=\xi  \\
		\end{array}
		\\
		\hline
		{1 \bar{1}}_- & 
		\begin{array}{cc}
			\rho _{1 \bar{1},1 \bar{1}}^{0 \bar{0}}=1 & \rho _{1 \bar{1},1 \bar{1}}^{2 \bar{2}}=-1 \\
		\end{array}
		\\
		\hline
		{1 \bar{1}}_+ & 
		\begin{array}{cc}
			\rho _{1 \bar{1},1 \bar{1}}^{0 \bar{0}}=1 & \rho _{1 \bar{1},1 \bar{1}}^{2 \bar{2}}=1 \\
		\end{array}
		\\
		\hline
		0 \bar{1}\oplus 2 \bar{1} & 
		\begin{array}{cccc}
			\rho _{0 \bar{1},0 \bar{1}}^{0 \bar{0}}=1 & \rho _{2 \bar{1},2 \bar{1}}^{0 \bar{0}}=1 & \rho _{0 \bar{1},2 \bar{1}}^{2 \bar{2}}=\frac{1}{\xi } & \rho _{2 \bar{1},0 \bar{1}}^{2 \bar{2}}=\xi  \\
		\end{array}
		\\
		\hline
		1 \bar{0}\oplus 1 \bar{2} & 
		\begin{array}{cccc}
			\rho _{1 \bar{0},1 \bar{0}}^{0 \bar{0}}=1 & \rho _{1 \bar{2},1 \bar{2}}^{0 \bar{0}}=1 & \rho _{1 \bar{0},1 \bar{2}}^{2 \bar{2}}=\frac{1}{\xi } & \rho _{1 \bar{2},1 \bar{0}}^{2 \bar{2}}=\xi  \\
		\end{array}
		\\
		\hline
	\end{array},
\end{equation}
where $\xi$ is a free parameter to choose. We set $\xi =1 $ in the following. Will use a subscript $\pm$ to distinguish the third and the fourth objects who have the same $V=1\bar{1}$.

The group $\mathbb{Z}_2$ is reconstructed as
\begin{equation}\label{eq:reconstructZ2Ising}
	G=\Aut A=\{\id_{0\bar{0}} \pm \id_{2\bar{2}}\}.
\end{equation}
For simplicity, we denote the elements by $\pm$.

These six simple objects in $\Rep A$ are graded by $G=\mathbb{Z}_2=\{+,-\}$. The first four are graded by $g=+$, while the last two by $g=-$.

The group action $T_g$ maps ${1\bar{1}}_{\pm}\to {1\bar{1}}_{\mp}$, and any other object to itself. The quivariant structure $u_+=\id_V$ and $u_-$ is given by
\begin{equation}\label{eq:equivariantUgIsing}
	\begin{array}{|c|c|}
		\hline
		0 \bar{0}\oplus 2 \bar{2} & u_- = \id_{0\bar{0}} - \id_{2\bar{2}}
		\\
		\hline
		0 \bar{2}\oplus 2 \bar{0} & u_- = \id_{0\bar{2}} - \id_{2\bar{0}}
		\\
		\hline
		{1 \bar{1}}_- & u_- = \id_{1 \bar{1}}
		\\
		\hline
		{1 \bar{1}}_+ & u_- = \id_{1 \bar{1}}
		\\
		\hline
		0 \bar{1}\oplus 2 \bar{1} & u_- = \id_{0\bar{1}} - \id_{2\bar{1}}
		\\
		\hline
		1 \bar{0}\oplus 1 \bar{2} & u_- = \id_{1\bar{0}} - \id_{1\bar{2}}
		\\
		\hline
	\end{array}
\end{equation}

As a result, the group action $T:\underline{G}\to \Aut_\ox \mathcal{C}$ has the trivial tensor structure $\id_V:T_{gh}V\to T_gT_h V$. The functor $T_g$ also has the trivial tensor structure $\id_{V\aox W}:T_g(V\aox W)\to T_g V\aox T_g W$.

\section{Charge condensation in Twisted quantum double phases}\label{sec:chargeCond}
\label{sec:chargeCondensationTQD}

Consider a twisted quantum double phase\cite{Hu2012a} characterized by $\mathcal{C}=\mathrm{Rep}(D^\omega G)$, the representation category of the quantum double $D^\omega G$, with $G$ a finite group and $\omega\in H^3[G,U(1)]$. For the definition of $D^\omega G$ and a brief review of its representations, see Appendix \ref{sec:tqdAndreps}.

Consider a charge condensation, in which all the charges condense. We shall first concretely construct a Frobenius algebra $A\in \mathcal{C}$ for such a condensation, and then derive the coherent states and the symmetry transformations $U_g$ defined previously.

The dual algebra $(D^\omega G)^*$ of $D^\omega G$ is a representation of $D^\omega G$, with the representation matrix given by
\begin{equation}\label{eq:VhRepTQD}
	\pi^R(\qab{g}{h})\kab{g'}{h'}=\ket{\qab{g'}{h'}\qdB^h(\qdA^x)^{-1}}=\delta_{h,h'}\frac{1}{\beta_{gh\ginv{g}}(g,\ginv{g})}\ket{\qab{g'\ginv{g}}{gh'\ginv{g}}},
\end{equation}
where $(\qdA^x)^{-1}$ is defined by $\qdA^x(\qdA^x)^{-1}\qdB^y=\qdB^y$ for all $y\in G$.

We consider two morphisms of the representations of $D^\omega G$: the coproduct $\Delta: (D^\omega G)^* \to (D^\omega G)^* \otimes (D^\omega G)^*$
and its transpose $\Delta^*: (D^\omega G)^* \otimes (D^\omega G)^*  \to (D^\omega G)^*$. The transpose $\Delta^*$ is defined by
\begin{equation}\label{eq:dualMultiplication}
	\Delta^*:\kab{g_1}{h_1}\otimes \kab{g_2}{h_2}\mapsto \delta_{g_1,g_2}\mu_{g_1}(h_1,h_2)\kab{g_1}{h_1h_2},
\end{equation}
where $\Delta$ is defined in Eq. \eqref{eq:coproductTQD} and $\mu$ is the coproduct structure defined in Eq. \eqref{eq:tqdMu}. 

Let
\begin{equation}\label{eq:VhDefinition}
	V_h=\text{vector space spanned by }\{\qdB^{h}\qdA^g\}_{g\in G}.
\end{equation}
Obviously, $(V_h,\pi^R)$ is a representation of $D^\omega G$. We shall use them to construct the Frobenius algebra $A$ and $\Rep A$.

Let
\begin{equation}\label{eq:frobeniusAlgebraInTQD}
	A=V_{h=1}.
\end{equation}
Then $(A,\Delta^*)$ is an associative algebra. Moreover, it carries a commutative Frobenius algebra structure, given by
\begin{equation}\label{eq:FAmultiplication}
	\Delta^*:A\ox A\to A, 
	\quad 
	\skab{g}{1}\ox \skab{g'}{1}\mapsto \delta_{g,g'}\skab{g}{1},
\end{equation}
\begin{equation}\label{eq:FAcomultiplication}
	\Delta:A\to A\ox A,
	\quad
	\skab{g}{1}\mapsto \skab{g}{1}\ox \skab{g}{1},
\end{equation}
\begin{equation}\label{eq:FAunit}
	\epsilon^*:\mathbb{C}\to A, 
	\quad
	1\mapsto \sum_g \skab{g}{1},
\end{equation}
and
\begin{equation}\label{eq:FAcounit}
	\epsilon:A\to \mathbb{C},
	\quad
	\skab{g}{1}\mapsto 1
\end{equation}

By definition \eqref{eq:AutAGroupCopy}, a finite group $\Aut(A)$ can be reconstructed:
\begin{equation}\label{eq:GroupReconstructed}
	\Aut(A)=\{\pi^L(\qdA^g\qdB^1),g\in G\},
\end{equation}
where $\pi^L$ is defined to be
\begin{equation}\label{eq:VhRepLTQD}
	\pi^L(\qab{g}{h})\kab{g'}{h'}=\ket{\qab{g}{h}\qab{g'}{h'}}.
\end{equation}
Note that since any left multiplication commutes with a right multiplication, $\pi^L(\qdA^g\qdB^1)$ is a morphism of  $(A,\pi^R)\to(A,\pi^R)$. We shall not distinguish $G$ from $\Aut(A)$, where $\pi^L(\qdA^g\qdB^1)$ is identified with $g$ since it defines a regular representation of $G$ with $A$ being the representation space.

The multiplication restricted to $\Delta^*: A\ox V_h\to V_h$ is commutative:
\begin{equation}\label{eq:AVhVhCommutative}
	\AVhVhCommutative = \AVhVhCommutativeAA,
\end{equation}
which defines $(V_h,\Delta^*)$ as an $A$-module with $\Delta^*: A\ox V_h\to V_h$ being the module action. By dimension counting, $L_{\Rep A}=\{(V_h,\Delta^*)\}_{h\in G}$ is a complete set of (representatives of) the simple objects in $\Rep A$.

%

The tensor product carries a group structure:
\begin{equation}\label{eq:VgVhVgh}
	V_g\otimes_A V_h = V_{gh},
\end{equation} 
\begin{equation}\label{eq:gradingVg}
	\grad V_g = g.
\end{equation}
Computing $\grad(V_g)$ involves the braiding in $\Rep_{D^\omega(G)}$ defined by $c_{V_{h_1},V_{h_2}}=\Sigma_{V_{h_1},V_{h_2}}\circ (\pi^R\ox \pi^R)(R)$, where $\Sigma_{(V_{h_1},,\pi^R),(V_{h_2},,\pi^R)}$ is the flip $v\ox w\mapsto w\ox v$, and $R$ is the universal $R$ matrix in Eq. \eqref{eq:Rmatrix}.

For each $(V_h,\Delta^*)\in L_{\Rep A}$, we define $(V_h,\Delta^{*g})\in\Rep A$ by
\begin{equation}\label{eq:VhDeltag}
	\VhDeltag.
\end{equation}

Each $V_h$ carries an equivariant structure $u_g=\pi^L(\qdA^g):(V_x,\pi^R)\to (V_{gx\ginv{g}},\pi^R)$ as a representation morphism. The defining property of $u_g$ can be verified
\begin{equation}\label{eq:psigDefinitionVhiVhk}
	\psigDefinitionVhiVhk,
\end{equation}
where in the second equality we used the commutativity $\pi^L(\qdA^g)\circ \Delta^*=\Delta^*\circ \pi^L(\Delta(\qdA^g))$. 

Using Eq. \eqref{eq:tqdMultiplication}, we have
\begin{equation}\label{eq:psigPsihPsigh}
	u_{g_1}u_{g_2}=\beta_{h}(g_1,g_2)u_{g_1g_2}: V_h\to V_{g_1g_2hg_1^{-1}g_2^{-1}}.
\end{equation}

The group action $T:\underline{G}\to \Aut_{\ox}(\mathcal{C})$ is given as follows. Define the functor $T_g:\Rep A\to\Rep A$ by $T_g(V_h,\Delta^*)=(V_h,\Delta^{*\ginv{g}})$ and $T_g(f)=f$. The tensor structure on $T_g$ is defined by
\begin{equation}\label{eq:TensorStructureOnTgMu}
	\mu_g(x,y)\id_{V_x\aox V_y}: T_g(V_x\aox V_y) \to T_{g}(V_x)\aox T_{g}(V_y).
\end{equation}

The tensor structure on the the functor $T_g$ is defined by
\begin{equation}\label{eq:TensorStructureTgamma}
	\beta_{x}(g,h)^{-1}\id_{V_x}:T_{gh}(V_x)\to (T_g\circ T_h)(V_x).
\end{equation}

In summary, we have shown that $\Rep A\cong Vect^{\omega}_G$, where $Vect^{\omega}_G$ is the category of finite dimensional $G$-graded vector spaces with an associativity $\omega\in H^3[G,U(1)]$. With the group action $T$, the equivariant structure $u_g=\pi^L(\qdA^{g})$.

\subsection{Effective model}\label{sec:TQDEffectiveModel}

We have explicitly constructed the group action $T$ and the equivariant structure $u_g$. Using these structures, we can write down the effective Hamiltonian. The Hilbert space is spanned by the basis
\begin{equation}
\ket{\qba{1}{g_0}; \qba{h_1}{g_1},\qba{h_2}{g_2},\dots,\qba{h_N}{g_N} }.
\end{equation}
The effective Hamiltonian is
\begin{equation}
H= -\sum_n Q_n-\lambda\sum_n W_n,
\end{equation}
where
\begin{equation}
Q_n\ket{\qba{1}{g_0}; \qba{h_1}{g_1},\dots,\qba{h_N}{g_N} }
=\delta_{h_n,1}\ket{\qba{1}{g_0}; \qba{h_1}{g_1},\dots,\qba{h_N}{g_N} },
\end{equation}
and
\begin{equation}
W_n\ket{\qba{1}{g_0}; \qba{h_1}{g_1},\dots,\qba{h_N}{g_N} }
=\delta_{g_n,g_0}\ket{\qba{1}{g_0}; \qba{h_1}{g_1},\dots,\qba{h_N}{g_N} }.
\end{equation}

The ground states are $\ket{\qba{1}{g}; \qba{1}{g},\dots,\qba{1}{g} }$. The actual ground state with a fixed $g$ will be chosen randomly by the system.

According to Definition \eqref{eq:globalSymmetryTransform} of the global symmetry transformation $U_g$, we have
\begin{equation}
U_g\ket{\qba{1}{g_0}; \qba{h_1}{g_1},\dots,\qba{h_N}{g_N} }=\pi^R\left(\Delta^N(\qdA^{g})\right)\ket{\qba{1}{g_0}; \qba{h_1}{g_1},\dots,\qba{h_N}{g_N} },
\end{equation}
where $\Delta^N=(\Delta\otimes \mathrm{id}^{\otimes{(N-1)}})\circ(\Delta\otimes \mathrm{id}^{\otimes{(N-2)}})\circ\dots \circ(\Delta\otimes \mathrm{id})\circ\Delta$.

\subsection{Example: $G=S_3$}

The simplest non-Abelian group is $G=S_3 = (r,s|r^3,s^2)$. There are three irreducible representations denoted by $j_0$, $j_1$ and $j_2$, with dimensions $\dim j_0=1$, $\dim j_1=1$, and $j_2=2$.

The twisted quantum double $D^\omega S_3$ has eight irreducible representations. See Section \ref{sec:repTQD} for the detailed construction of these representations. They are classified by the three conjugacy classes of $S_3$, denoted by $C_1=\{1\}$, $C_2=\{r,r^2\}$, and $C_3=\{s,rs,r^2s\}$. We denote by $j_0,j_1,j_2$ the irreducible projective representations of $Z_{1}=S_3$, by $\theta_0,\theta_1,\theta_2$ of $Z_{r}=\mathbb{Z}_3$, and by $\mu_0,\mu_1$ of $Z_{s}=\mathbb{Z}_2$, where $Z_h=\{g\in G|gh=hg\}$ is the centralizer. The eight irreducible representations of $D^\omega G$ are labeled by $(C_1,j_0)$, $(C_1,j_1)$, $(C_1,j_2)$, $(C_2,\theta_0)$, $(C_2,\theta_1)$, $(C_2,\theta_2)$, $(C_3,\mu_0)$, and $(C_3,\mu_1)$, respectively. 

The Frobenius algebra $A$ and $\Rep A$ can be expressed in terms of these irreducible representations. Let $A=(C_1,j_0)\oplus (C_1,j_1)\oplus (C_1,j_2)\oplus (C_1,j_2)$. Label the vectors in $A$ by $\ket{alpha;j,\beta}$, where $\alpha$ is the multiplicity index of $j$ appearing in $A$, and $\ket{j,\beta}$ is the basis vectors of the representation space of $j$.

There are six simple objects in $\mathrm{Rep}A$:
\begin{equation}\label{eq:SixObjectInRepAForS3}
	\begin{aligned}
		&V_1=A=(C_1,j_0)\oplus (C_1,j_1)\oplus (C_1,j_2)\oplus (C_1,j_2),\\
		&V_r=(C_2,\theta_0)\oplus (C_2,\theta_1)\oplus (C_2,\theta_2),\\
		&V_{r^2}=(C_2,\theta_0)\oplus (C_2,\theta_1)\oplus (C_2,\theta_2),\\
		&V_{s}=(C_3,\mu_0)\oplus (C_3,\mu_1),\\
		&V_{rs}=(C_3,\mu_0)\oplus (C_3,\mu_1),\\
		&V_{r^2s}=(C_3,\mu_0)\oplus (C_3,\mu_1).\\
	\end{aligned}
\end{equation}
The action $g:A\to A$ is given by
\begin{equation}
g\ket{\alpha;j,\beta}=\sum_{\alpha'}\rho^j_{\alpha'\alpha}(g)\ket{\alpha';j,\beta},
\end{equation}
where $\rho^j$ is the representation matrix of $g$. The global symmetry transforms under the representation $\pi^C_\mu$ defined in Eq. \eqref{eq:repTQD} of Appendix \ref{sec:repTQD}.

\section{Classification of maximal condensations in twisted quantum double phases}\label{sec:dyonCond}

We have discussed charge condensation in twisted quantum double phases with a finite group $G$. In this section, we consider more generic condensation in a twisted quantum double phase where the condensation anyons are dyons  (charge-flux composite). We consider the cases with a maximal ($\dimc A=|G|$) absorbing commutative Frobenius algebra. We call such condensation a maximal bosonic condensation. The resulting condensed phase is an SPT phase.

\subsection{Quantum double phases}

Let $\scC=\Rep(D(G))$ with $G$ a finite group. The complete set of (representatives of) simple objects is denoted by 
\begin{equation}\label{eq:setQD}
	L_{\Rep(D(X))}=\{(a,\alpha)\}
\end{equation}
where $a$ is the representative of every conjugacy class of $G$, and $\alpha$ is an irreducible representation of the centralizer $Z_a=\{z\in G|za=az\}$. The $S$-matrix and the twist \cite{Coste2000} are
\begin{equation}\label{eq:smatrixQD}
	S_{(a,\alpha),(b,\beta)}=\frac{|G|}{|Z_a| |Z_b|}\sum_{g\in G(a,b)}\bar\chi_\alpha(gb\ginv{g})\bar\chi_\beta(\ginv{g}ag)
\end{equation}
\begin{equation}\label{eq:twistQD}
	\theta(a,\alpha)=\frac{\chi(a)}{\dim_\alpha}
\end{equation}
where
\begin{equation}\label{eq:GabQD}
	G(a,b)=\{g\in G|agb\ginv{g}=gb\ginv{g}a\}.
\end{equation}

The maximal bosonic condensations are classified by the absorbing commutative Frobenius algebras with with $\dimc A=|G|$.  The full subcategory generated by such $A$ is a Lagrangian subcategory of $\scC$. (A Lagrangian subcategory is a symmetric subcategory $\scS\subset \scC$ such that every $U\in\scC$ but $U\notin \scS
$ has a nontrivial monodromy with $\scS$).

The Lagrangian subcategories of $\scC=\Rep(D(G))$ are characterized by $(N,B)$, where $N\subset G$ is an Abelian normal subgroup of $G$, and $B:N\times N\to U(1)$ is a $G$-invariant alternating bicharacter on $N$, satisfying
\begin{equation}\label{eq:alternatingBicharacterBQD}
	\begin{aligned}
		&B(h_1h_2,h)=B(h_1,h)B(h_2,h),\\
		&B(h,h_1h_2)=B(h,h_1)B(h,h_2),\\
		&B(h,h)=1,
	\end{aligned}
\end{equation}
for all $h,h_1,h_2\in N$, and the $G$-invariance condition
\begin{equation}\label{eq:BGinvariance}
	B(h_1,h_2)=B(gh_1\ginv{g},gh_2\ginv{g}).
\end{equation}

Note that the above defining properties infer that $B(h_1,h_2)B(h_2,h_1)=1$.

Explicitly, the Lagrangian subcategory $\mathcal{L}_{(N,B)}$ is generated by
\begin{equation}\label{eq:lagrangianRepsQD}
	L_{\mathcal{L}_{(N,B)}}=\left\{
	(a,\alpha)\in L_{\Rep(D(G))}\left|
	a\in N, \chi_\alpha(h)=B(a,h)\dim_{\alpha}, \forall h\in N
	\right.
	\right\}
\end{equation}

The Frobenius algebra $A$ is then
\begin{equation}
A=\bigoplus_{j\in L_{\mathcal{L}_{(N,B)}}}j^{\oplus\dimc(j)}.
\end{equation}

The proof is given in Ref\cite{Naidu2008}, which is sketched in what follows.

The subcategory $\mathcal{L}_{(N,B)}$  is Lagrangian if $\mathcal{L}_{(N,B)}$ is symmetric with $\theta_{(a,\alpha)}=1$ for every $(a,\alpha)\in L_{\mathcal{L}_{(N,B)}}$, and $\dim \mathcal{L}_{(N,B)}=|G|$. We first rewrite the condition that  $\mathcal{L}_{(N,B)}$ is symmetric. By Ref\cite{Muger2003a}, we need to verify that for all $(a,\alpha),(b,\beta)$, the $S$-matrix satisfies
\begin{equation}\label{eq:symmetricSmatrix}
	S_{(a,\alpha),(b,\beta)}=\dim(a,\alpha)\dim(b,\beta)=|C_a|\dim_\alpha |C_b|\dim_\beta.
\end{equation}
Plug Eq. \eqref{eq:alternatingBicharacterBQD} in, Eq. \eqref{eq:symmetricSmatrix} is equivalent to the following two conditions:
\begin{enumerate}
	\item The conjugacy classes $C_a$ and $C_b$ commute element-wise.
	\item $\chi_\alpha(gb\ginv{g})\chi_\beta(\ginv{g}ag)=\dim_{\alpha}\dim_\beta$, for all $g\in G$.
\end{enumerate}

We check that $L_{\mathcal{L}_{(N,B)}}$ satisfies these two conditions. The first one is straightforward since $N$ is an Abelian normal subgroup. The second can be verified directly using the defining properties of $B$:
\begin{equation}\label{eq:SymmetricSecondCondition}
	\chi_\alpha(gb\ginv{g})\chi_\beta(\ginv{g}ag)=B(a,gb\ginv{g})B(gb\ginv{g},a)\dim_\alpha\dim_\beta=\dim_\alpha\dim_\beta.
\end{equation}

Also, we compute $\theta_{(a,\alpha)}=\frac{\chi_\alpha(a)}{\dim_{\alpha}}\dim_{\alpha}=1$.

Finally,
\begin{equation}\label{eq:dimLQD}
	\begin{aligned}
		\dim{\mathcal{L}_{(N,B)}}=&\sum_{(a,\alpha)\in L_{\Rep(D(G))}}\dim(a,\alpha)^2
		\\
		=&\sum_{a\in N\cap R}|C_a|^2\sum_{\alpha:\chi_\alpha(h)=B(a,h)\dim_\alpha\forall h\in N}\dim_\alpha^2
		\\
		=&\sum_{a\in N\cap R}|C_a|^2\frac{|Z_a|}{|N|}
		=|G|,
	\end{aligned}
\end{equation}
where in the third equality, use is made of Clifford's theorem and Frobenius reciprocity (see Ref\cite{Naidu2008} for a detailed proof), and the last equality holds since $N$ is a normal subgroup.

\subsection{Twisted quantum double phases}

The complete set of (representatives of) simple objects in $\scC=\Rep(D^{\omega}(G))$ is denoted by 
\begin{equation}\label{eq:setTQD}
	L_{\Rep(D^\omega(G))}=\{(a,\alpha)\},
\end{equation}
where $a$ is the representative of every conjugacy class of $G$, and $\alpha$ is an irreducible $\beta_a$-representation of the centralizer $Z_a=\{z\in G|za=az\}$. By a $\beta_a$-representation we mean a projective representation such that $\rho(z)\rho(z')=\beta_a(z,z')\rho(zz')$.

The $S$-matrix and the twist \cite{Coste2000} are
\begin{equation}\label{eq:smatrixTQD}
	\begin{aligned}
		&S_{(a,\alpha),(b,\beta)}
		\\
		=&\sum_{g\in C_a,g'\in C_b:gg'=g'g}\overline{\left(\frac{\beta_a(x,g')\beta_a(xg',x^{-1})\beta_b(y,g)\beta_b(yg,y^{-1})}{\beta_b(x,x^{-1})\beta_b(y,y^{-1})}\right)}\bar\chi_\alpha(xg'x^{-1})\bar\chi_\beta(ygy^{-1})
	\end{aligned}
\end{equation}
and 
\begin{equation}\label{eq:twistTQD}
	\theta(a,\alpha)=\frac{\chi(a)}{\dim_\alpha},
\end{equation}
where $x,y$ are uniquely determined by $g=x^{-1}ax,g'=y^{-1}by$.

Similar to the untwisted case, the Lagrangian subcategories of $\scC=\Rep(D(G))$ are characterized by $(N,B)$, where $N\subset G$ is an Abelian normal subgroup of $G$, and $B:N\times N\to U(1)$ is a $G$-invariant alternating $\omega$-bicharacter on $N$, satisfying
\begin{equation}\label{eq:alternatingBicharacterBTQD}
	\begin{aligned}
		&B(h_1h_2,h)=\beta_h(h_1,h_2)B(h_1,h)B(h_2,h),\\
		&\beta_h(h_1,h_2)B(h,h_1h_2)=B(h,h_1)B(h,h_2),\\
		&B(h,h)=1,
	\end{aligned}
\end{equation}
for all $h,h_1,h_2\in N$, and the $G$-invariance condition
\begin{equation}\label{eq:BGinvarianceTQD}
	B(x^{-1}ax,h)=\frac{\beta_a(x,h)\beta_a(xh,x^{-1})}{\beta_a(x,x^{-1})}B(a,xhx^{-1}),
\end{equation}
for all $x\in G,a\in N\cup R,h\in N$.

The Lagrangian subcategory $\mathcal{L}_{(N,B)}$ is generated by
\begin{equation}\label{eq:lagrangianRepsTQD}
	L_{\mathcal{L}_{(N,B)}}=\left\{
	(a,\alpha)\in L_{\Rep(D^{\omega}(X))}\left|
	a\in N, \chi_\alpha(h)=B(a,h)\dim_{\alpha}, \forall h\in N
	\right.
	\right\}.
\end{equation}

The proof\cite{Naidu2008} is similar to that in the untwisted case. For $(a,\alpha),(b,\beta)\in L_{\mathcal{L}_{(N,B)}}$, one can show that
\begin{equation}\label{eq:symmetricSmatrixTQD}
	S_{(a,\alpha),(b,\beta)}=\dim(a,\alpha)\dim(b,\beta)=|C_a|\dim_\alpha |C_b|\dim_\beta,
\end{equation}
which verifies $\theta_{(a,\alpha)}=\frac{\chi_\alpha(a)}{\dim_{\alpha}}\dim_{\alpha}=1$, and then computes
\begin{equation}\label{eq:dimLTQD}
	\dim{\mathcal{L}_{(N,B)}}
	=|G|.
\end{equation}

\subsection{Global symmetry and Partial Electric-Magnetic duality}

We will show that every maximal bosonic condensation in a twisted quantum double phase is equivalent to a pure charge condensation.

Let $\mathcal{C}=\mathrm{Rep}(D^\omega G)$ be the modular tensor category characterizing a twisted quantum double phase. For every maximal bosonic condensation in $\mathcal{C}$, there is a Lagrange subcategory $\mathcal{L}_{(N,B)}$ with a normal Abelian subgroup $N$ and an alternating function $B$.

The hidden global symmetry $X=\Aut(A)$ is different from $G$ if $N\neq \{1\}$. Then $\mathcal{L}_{(N,B)}=\Rep_X$. The group structure of $X$ is determined by the following properties:
\begin{enumerate}
	\item The $G$ can be written as a semidirect product $G=N\rtimes_{F} K$, where $F\in H^{2}(K,N)$ is a $2$-cocycle. The 3-cocycle can be expressed\cite{UribeJongbloed2017} explicitly (up to a coboundary) as
\begin{equation}
	\label{eq:omegaFepsilon}
	\omega ((a_{1} ,k_{1} ),(a_{2} ,k_{2} ),(a_{3} ,k_{3} ))=\hat{F} (k_{1} ,k_{2} )(a_{3} )\epsilon (k_{1} ,k_{2} ,k_{3} ),
\end{equation}
where $\hat{F} \in H^{2}( K,\hat{N})$ is a 2-cocycle, and  $\epsilon \in C^{3}( K,U( 1))$ satisfies $\delta _{K} \epsilon =\hat{F} \land F$, i.e., 
\begin{equation}
	\label{eq:FFcondition}
	\delta_K\epsilon ( k_{1} ,k_{2} ,k_{3} ,k_{4}) 
	=\frac{\epsilon ( k_{2} ,k_{3} ,k_{4}) \epsilon ( k_{1} ,k_{2} k_{3} ,k_{4}) \epsilon ( k_{1} ,k_{2} ,k_{3} ) }{\epsilon ( k_{1} k_{2} ,k_{3} ,k_{4}) \epsilon ( k_{1} ,k_{2} ,k_{3} k_{4}) }
	=\hat{F}( k_{1} ,k_{2})( F( k_{3} ,k_{4})).
\end{equation}

	\item The $X$ can be written as $X=K\ltimes_{\hat{F}}\hat{N}$. Here $\hat{N}$ is the Abelian group whose elements are the unitary irreducible representations of $N$.
\end{enumerate}

Define $\omega'\in H^3(G',U(1))$ by
\begin{equation}
	\label{eq:alphaHatFepsilon}
	\omega'((x_{1} ,\rho _{1} ),(x_{2} ,\rho _{2} ),(x_{3} ,\rho _{3} ))=\rho _{1} (F(x_{2} ,x_{3} ))\epsilon (x_{1} ,x_{2} ,x_{3} ).
\end{equation}
Then $\mathrm{Rep}(D^\omega G)\cong \mathrm{Rep}(D^{\omega'} X)$(\cite{Naidu2011,UribeJongbloed2017}).

According to Ref\cite{Hu2020}, such an equivalence is realized by a partial electric-magnetic duality, under the Fourier transform $N\to \hat{N}$. The condensate anyons form the complete set $L_{\mathcal{L}_{(N,B)}}$ of composites of pure $N$-flux and $K$-charge. Under the partial electric-magnetic duality, they become pure $(K\ltimes_{\hat{F}}\hat{N})$-charges. Indeed, the condensate anyons form the set
\begin{equation}\label{eq:lagrangianRepresentative}
	L_{\mathcal{L}_{(N,B)}}=\left\{
	(a,\alpha)\in L_{\Rep(D^{\omega}(X))}\left|
	a\in N, \chi_\alpha(h)=\dim_{\alpha}, \forall h\in N
	\right.
	\right\}.
\end{equation}
The condensate anyons range over all possible all $N$-fluxes $a\in N$, and contains no $N$-charge according to the condition $\chi_\alpha(h)=\dim_{\alpha}$. Hence they are the set of composites of pure $N$-flux and $K$-charge.

The global symmetry transformation can be derived from the EM duality. The isomorphism\cite{Hu2020} $D^{\alpha } G\rightarrow D^{\alpha '}  X$ is given by
\begin{equation}\label{eq:fourierTransformTQD}
A^{(a,x)} B^{(b,y)} \mapsto \frac{1}{|N|}\sum _{\rho ,\eta \in \hat{N}}\rho (a)\overline{\eta (b)} \ \tilde{B}^{\left( xy\bar{x} ,\rho \right)}\tilde{A}^{(x,\eta )},
\end{equation}
where $(a,x),(b,y)$ are group elements of $G$ and $\left( xy\bar{x} ,\rho \right),(x,\eta )$ of $X$.

Hence the group elements of $X$ is expressed in terms of elements of $D^{\omega'}X$:
\begin{equation}
\tilde{\qdA}^{(x,\eta)}=\frac{1}{|N|}\sum_{b\in N,y\in K}\eta(b)\qab{(1,x)}{(b,y)}.
\end{equation}
The global symmetry transformation and the effective theory are constructed in a similar way as in Section \ref{sec:chargeCondensationTQD}.

\subsection{Example for $G=D_{m}$}

The simplest non-Abelian group is the symmetric group $ S_{3}$, or the Dihedral group $ D_{3}$. Denote the elements of $ D_{3}$ by $ ( A,a)$, for $ A=0,1$ and $ a=0,1,\dotsc ,m$, where $ m=3$ for $ D_{3}$ case. The multiplication is given by
\begin{equation}
	( A,a) \circ ( B,b) =\left( \langle A+B\rangle _{2} ,\langle ( -1)^{B} a+b\rangle _{m}\right)
\end{equation}
where $ \langle n\rangle _{m}$ stands for $ n\ \mathrm{mod} \ m$.

For convenience, we rewrite $ ( A,a)$ as $ A_{a}$ for short. Define the 3-cocycle $ \omega $ by

\begin{equation}
	\omega ( A_{a} ,B_{b} ,C_{c}) =\exp\left\{\frac{2\pi i}{m^{2}}\left[( -1)^{B+C} a\left(( -1)^{C} b+c-\langle ( -1)^{C} b+c\rangle _{m}\right) +\frac{m^{2}}{2} ABC\right]\right\}.
\end{equation}
It generates 3-cocycle representatives $ \omega ^{k}$ for all other cohomology classes $ k=0,1,\dotsc ,2m-1$.

For $ m=3$, there are 8 irreducible representations of $ D^{\omega }( G)$, namely

\begin{center}
	
	\begin{tabular}{c|ccc}
		\hline 
		$ K_{1} =\{0_{0}\}$ & $ \rho _{0}$ & $ \rho _{1}$ & $ \rho _{2}$ \\
		\midrule 
		$ K_{2} =\{0_{1} ,0_{2}\}$ & $ \mu _{0}$ & $ \mu _{1}$ & $ \mu _{2}$ \\
		\hline 
		$ K_{3} =\{1_{0} ,1_{1} ,1_{2}\}$ & $ \nu _{0}$ & $ \nu _{1}$ &  \\
		\bottomrule
	\end{tabular}
\end{center}
where $ \rho _{0} ,\rho _{1} ,\rho _{2}$ are the three irreducible representations of the centralizer $ Z_{e} =S_{3}$, $ \mu _{0} ,\mu _{1} ,\mu _{2}$ are the three irreducible projective representations of $ Z_{0_{1}} =\mathbb{Z}_{3}$ (with $ \mu _{0}$ being the trivial one), and $ \nu _{0} ,\nu _{1}$ are the two irreducible projective representations of $ Z_{1_{0}} =\mathbb{Z}_{2}$ (with $ \nu _{0}$ being the trivial one).

The Abelian normal subgroups $ N$ of $ S_{3}$ is either $ K_{1} =\{0_{0}\}$ or $ K_{1} \sqcup K_{2} =\{0_{0} ,0_{1} ,0_{2}\}$. For $ N=K_{1} =\{0_{0}\}$, the Lagrangian subcategory is generated by $ \rho _{0} ,\rho _{1} ,\rho _{2}$, i.e., all irreducible representations of $ S_{3}$. The corresponding Frobenius algebra is $ A=\rho _{0} \oplus \rho _{1} \oplus \rho _{2} \oplus \rho _{2}$.

For $ N=K_{1} \sqcup K_{2} =\{0_{0} ,0_{1} ,0_{2}\}$, the solution to the alternating bicharacter exits iff $ \omega =1$. The only solution is $ B( h_{1} ,h_{2}) =1$ for all $ h_{1} ,h_{2} \in N$. Hence \ are those \ From $ \rho _{0} ,\rho _{1} ,\rho _{2} ,\mu _{0} ,\mu _{1} ,\mu _{2}$, we pick up $ \rho _{0} ,\rho _{1} ,\mu _{0}$ as the generating objects of the Lagranginan subcategory, which satisfying $ \rho ( h) =\mathrm{id}$ for $ h\in N$. Hence the Frobenius algebra is $ A=\rho _{0} \oplus \rho _{1} \oplus \mu _{0} \oplus \mu _{0}$. The group $X=\mathrm{Aut}( A)$ turns out to be $ X=D_{3}$ (In fact, $ D_{3}$ is the only group that can have irreducible representations $ \rho _{0} ,\rho _{1} ,\mu _{0}$). Note that $ X=D_{3}$ is different from $ G=D_{3}$. The anyon condensation in the phase $ \mathrm{Rep}( D( G))$ for $ G=D_{3}$ will results in a SPT phase with global symmetry $ X=D_{3}$.

The equivalence $ \mathrm{Rep}( D( D_{4})) \cong \mathrm{Rep}\left( D^{\omega '}(\mathbb{Z}_{2} \times \mathbb{Z}_{2} \times \mathbb{Z}_{2})\right)$ is known widely for certain $\omega'\in H^3[\mathbb{Z}_{2} \times \mathbb{Z}_{2} \times \mathbb{Z}_{2},U(1)]$. All other normal Abelian subgroups of $ \mathrm{Rep}( D( D_{4}))$ give rise to a global symmetry group $ X=D_{4}$.

In general, for $ G=D_{m}$ with $ m\geqslant 5$, the normal Abelian subgroups of $ G$ are precisely the rotation subgroups $ \mathbb{Z}_{m/k}$ generated by $ 0_{k}$, where $ k$ is a divisor of $ m$. The Lagrangian subcategories are equivalent to $ \mathrm{Rep}( Dih( N))$, where $ N=\mathbb{Z}_{m/k} \times \mathbb{Z}_{k}$, and $ Dih( N) =N\rtimes \mathbb{Z}_{2}$ is the generalized Dihedral group for an Abelian group $ N$, with $ \mathbb{Z}_{2}$ acting on $ N$ by inverting elements. The anyon condensation will result in a SPT phase, with a global symmetry $ X=\mathbb{Z}_{m/k} \times \mathbb{Z}_{k}$.

\section{Restore hidden symmetry: gauging symmetry enriched topological phases}\label{sec:gauging}

We have discussed the effective theory of complete anyon condensation. The condensed phase is described by coherent states using $(\Rep A)^G$ (in terms of the simple objects in $\Rep A$).

We consider a reversed process of condensation: How can we restore the hidden symmetry? We can introduce a gauge invariance term in the Hamiltonian:
\begin{equation}
-J\sum_n T_n,
\end{equation}
where $J>0$ and $T_n$ acts on the basis vectors \eqref{eq:psigBasisRepA} by
\begin{equation}
T_n\ket{g_NV_N}
=\frac{1}{|G|}\sum_g\left(u_g+u_g^{\dagger}\right)\ket{g_NV_N}.
\end{equation}
The $T_n$ does not commute with $W_n$. When $J\gg\lambda$, the energy eiegenstates have basis vectors in Eq. \eqref{eq:quasiparticleStates}.

The topological charge of anyons are no longer simple objects in $\Rep A$ but those in $(\Rep A)^G$. It is known\cite{Kirillov2002,Muger2004}, that $(\Rep A)^G\cong \mathcal{C}$; hence, we restore the parent phase characterized by $\mathcal{C}$.


\acknowledgments

YT thanks Yong-Shi Wu for the inspiring discussion about the condensed states and the symmetry arrangement. YW is supported by NSFC grant No. 11875109, General Program of Science and Technology of Shanghai No. 21ZR1406700, Fudan University Original Project (Grant No. IDH1512092/009). LYH acknowledges the support of NSFC (Grant No.
11922502, 11875111). YW and LYH are also supported by the Shanghai Municipal Science and Technology Major Project (Grant No.2019SHZDZX01).

\appendix

\section{Frobenius algebra and its modules}

This appendix consists some preliminary mathematics of a Frobenius algebra and its modules. To be clear, instead of containing complete proofs for every statement, we just intend to provide enough computational tools to support the paper. 


We start with strict tensor categories. The standard definitions of tensor categories can be found in cf. the standard reference\cite{MacLane1998}. A tensor category is strict if its tensor product satisfies associativity $(X\ox Y)\ox Z= X\ox (Y\ox Z)$ on the nose and its unit object $\unit$ satisfies $X\ox \unit=\unit\ox X=X$ for all $X$. Throughout the paper, we assume the categories over $\mathds{C}$ are always Abelian.




In our convention, we represent morphisms in a tensor category by tangle diagrams. These diagrams are drawn to be read from the bottom to the top. With 
$a: X\rarr Y, b: Y\rarr Z, c: U\rarr V, d: V\rarr W$, we represent
\[ b \otimes d \circ a\otimes c = (b\circ a)\otimes (d\circ c)\in\Hom(X\otimes U, Z\otimes W) 
\]
by
\[
\graphNotationAA.
\]

The $Y$ is a subobject of $X$, if there are morphisms $e:X\rarr Y,f:Y\rarr X$ such that $e\circ f=\id_X$, and $p=f\circ e\in \mathrm{End}(X)$ is an idempotent, i.e., $p^2=p$. We denote $Y$ arising from the $X$ and $p$ by $\obj X p$. If every idempotent in a category $\C$ arises in this way, we say `idempotents split in $\C$' or `$\C$ has subobjects'.

\subsection{Frobenius algebra}\label{sec:frobeniusAlgebraAndModules}

Assume $\C$ is a modular tensor category. So $\C$ has two-sided duals.  This means, for every $X\in \C$, there exists a $X^*\in \C$ together with rigidity morphisms $e_X:\unit \rarr X\ox X^*, d_X:X^*\ox X\rarr \unit,\varepsilon_X:\unit \rarr X^*\ox X,\eta_X:X\ox X^*\rarr \unit$, depicted by the diagrams:
\begin{equation}
	\label{eq:twoSidedDuals}
	e_X = \twoSidedDual,
	\quad
	d_X = \twoSidedDualAA,
	\quad
	\varepsilon_X = \twoSidedDualAB,
	\quad
	\eta_X = \twoSidedDualAC,
\end{equation}
satisfying the equations
\begin{equation}
	\twoSidedDualEqs,
\end{equation}
where a downward arrow denotes $X^*$, and an upward arrow denotes $X$.

We also assume $\C$ is pivotal. Thus, $\forall f:X\rarr Y$, we have
\begin{equation}\label{eq:pivotalCondition}
	\pivotalCondition.
\end{equation}

For every $X\in\C$, the quantum dimension is defined as
\begin{equation}\label{eq:defDimXInC}
	\dimc X=\defDimXInC.
\end{equation}
The last equality holds because $\C$ is spherical.

$\forall X\in \C$ the twist $\theta_X$ is defined as 
\begin{equation}\label{eq:defTwist}
	 \, = \, \defTwistAA\,=\,\defTwistAB.
\end{equation}
The second equality holds because $\C$ is pivotal and semisimple.

In categrory $\C$, one can define Forbenius algebras. In $\C$, a Frobenius algebra $A$ is a quintuple $(A,\mu,\Delta,\iota,\epsilon)$, where $A$ is an object and 
$\iota:\unit\rightarrow A, \epsilon: A\rightarrow\unit, \Delta: A\rightarrow A\ox A, \mu:A\ox A\rightarrow A$ 
are morphisms. These morphisms should satisfy the following conditions:
\begin{equation} \label{eq:FA1}
	\Delta\otimes\id_A\mcirc \Delta=\id_A\otimes \Delta\mcirc \Delta, 
\end{equation}
\begin{equation} \label{eq:FA2}
	\mu \mcirc \mu\otimes\id_A=\mu \mcirc \id_A\otimes \mu, 
\end{equation}
\begin{equation} \label{eq:FA3}
	\epsilon\otimes\id_A\mcirc \Delta = \id_A = \id_A\otimes \epsilon\mcirc \Delta,
\end{equation}
\begin{equation} \label{eq:FA4}
	\mu \mcirc \iota\otimes\id_A = \id_A = \mu \mcirc \id_A\otimes \iota,
\end{equation}
\begin{equation} \label{eq:FA5}
	\mu\otimes\id_A \mcirc \id_A\otimes \Delta = \Delta\mcirc \mu = \id_A\otimes \mu\mcirc \Delta\otimes \id_A. 
\end{equation} 

In the diagram language, these morphisms are depicted by 
\[
\Delta=\froAlgDelta,
\qquad
\mu=\froAlgMu,
\qquad
\epsilon=\froAlgEpsilon,
\qquad
\iota=\froAlgIota,
\]
where we use the red line to denote $A$.

In this language, the defining conditions \eqref{eq:FA1}\eqref{eq:FA2}\eqref{eq:FA3}\eqref{eq:FA4}\eqref{eq:FA5} are
\begin{equation}\label{eq:FA1'}
	\text{comoind: }
	\froAlgCondAA
\end{equation}
\begin{equation}\label{eq:FA2'}
	\text{monoid: }\,
	\froAlgCondAB
\end{equation}
\begin{equation}\label{eq:FA3'}
	\text{Frobenius condition: }\qquad\quad
	\froAlgCondAC
\end{equation}

Besides, there are a couple of additional conditions on $A$:
\begin{enumerate}
	\item  the uniqueness of unit
	\begin{equation}\label{eq:uniqueUnit}
		\dim \Hom_{\C}(\unit,A)=1,
	\end{equation}
	\item and $\theta_A=\id_A$.
\end{enumerate}
By the pivotal property,  $\theta_A=\id_A$ iff $\theta_X=\id_X$ for every direct summand $X$ of $A$.

The uniqueness of unit implies that $A$ is \textit{strongly separable}, i.e.,
\begin{equation}\label{eq:dimAPic}
	\froAlgSeparableAA
\end{equation}
where $\lambda_1,\lambda_2\in \End(\unit)=\mathds{C}$. Among these two factors, $\lambda_2$ can be normalized as $1$, and $\lambda_1$ is proven to be $\dimc A$:
\begin{equation}\label{eq:proveDimA}
	\proveDimA,
\end{equation}
where in the second equality we used $\dim \Hom_\C(\unit,A)=1$, and the last equality is due to 
\begin{equation}\label{eq:twistDimA}
	\twistDimA=\dim_\C A,
\end{equation}
which holds for semisimple $\mathcal{C}$.

The uniqueness of the unit and $\theta_A=\id_A$ implies that $A$ is \textit{commutative},  i.e.,
\begin{equation}\label{eq:commutativeAlg}
	\froAlgCommutative,
\end{equation}
where the crossing is the braiding $c_{A,A}:A\ox A\rarr A\ox A$ in $\C$. Cf Corollary 3.10 in Ref\cite{Fuchs2002}.


\subsection{$A$-modules}

A left module over $A$ in $\C$ is a pair $(V,\rho)$, where $V\in\C$ and $\rho:A\ox V\to V$ satisfies $\rho\mcirc (\rho\mcirc \id_A)=\rho\mcirc (\id_V\ox \mu)$ and $\rho\mcirc(\id_V\ox \iota)=\id_V$. In diagram language, $\rho$ is depicted as $\leftModuleAction$,
and those conditions become
\begin{equation}\label{eq:defModule}
	\leftModule.
\end{equation}

All the left $A$-modules form a category $\Rep A$ with the unit $(A,\mu)$. The morphisms (called $A$-morphisms) defined by
\begin{equation}\label{eq:Amorphisms}
	\Hom_{A}((V,\rho),(V',\rho'))=\left\{ t\in \Hom_{\C}(V,V')\Biggl.\Biggr |\quad\rightModuleMorphisms\quad\right\}.
\end{equation}
A $\mathcal{C}$-morphism $f:V\rarr V'$ in $\C$ is a $A$-morphism iff
\begin{equation}\label{eq:AmorphismExpectation}
	\AmorphismExpectation
\end{equation}


The $\Rep A$ can be extended to be a tensor category.  $\forall (V,\rho),(V',\rho')\in \Rep A$, their tensor product $(V,\rho)\ox(V',\rho')$ is defined as $(V\aox V',\rho\aox\rho')$, where $V\aox V'$ is a subobject of $V\ox V'$ defined by the idempotent $\pi:V\ox V'\rarr V\ox V'$:
\begin{equation}
	\idenpotentPi,
\end{equation}
and $\rho\aox\rho'$ is given by
\begin{equation}\label{eq:tensorRepA}
	\tensorRepA.
\end{equation}
Using Eq. \eqref{eq:dimAPic}, it is easy to see that $\pi^2=\pi$. Furthermore, one can verify that $\rho\aox\rho'$ is a module action when $A$ is commutative.

If $\theta_A=\id_A$, $\Rep A$ has two-sided duals and is pivotal. For every $(V,\rho)\in \Rep A$, there exists a dual $(V,\rho)^*:=(V^*,\rho^*)\in \rep A$, where $\rho^*:A\ox V^*\rarr V^*$ is given by
\begin{equation}
	\dualRho,
\end{equation}
with the corresponding rigidity morphisms $\bar e_V:A \rarr V\aox V^*, \bar d_V:V^*\aox V\rarr A,\bar \varepsilon_V:A \rarr V^*\aox V,\bar \eta_V:V\aox V^*\rarr A$ are defined by
\begin{equation}\label{eq:dualMorphismsRepA}
	\dualRepAMorphisms
\end{equation}

Using the pivotal condition \eqref{eq:pivotalCondition}, $\rho^*$ can be also written as
\begin{equation}
	\dualRhoAA.
\end{equation} 

When $\theta_A=\id_A$, we have well defined two-sided duals in $\Rep A$, i.e., the morphisms. \eqref{eq:dualMorphismsRepA} are $A$-morphisms:
\begin{equation}
	\dualRepA
\end{equation}

For every $(V,\rho)\in \Rep A$, we define its quantum dimension $\dim_A V$ by
\begin{equation}\label{eq:defDimVInRepA}
	\defDimVInRepA,
\end{equation}
since the left hand side is a morphism in $\Hom_{\Rep A}(A,A)=\mathds{C}$, $\dim_A V\in \mathds{C}$. Using Eq. \eqref{eq:dimAPic}, we see that
\begin{equation}\label{eq:evaluateDimAV}
	\dima V=\frac{\dimc V}{\dimc A}.
\end{equation}
The quantum dimension of the tensor product $(V\aox W,\rho_V \aox \rho_W)$ can be computed by 
\begin{equation}\label{eq:dimVaoxW}
	\dima(V\aox W)=\frac{\dim_\C(V\aox W)}{\dimc A}=\frac{\dimc V\dimc W}{(\dimc A)^2}=\dima V \dima W,
\end{equation}
where in the second equality we computed
\begin{equation}\label{eq:dimCVaoxW}
	\dimCVaoxW.
\end{equation}

\subsection{Simple objects in $\Rep A$}

In this part of the appendix, we explore some properties of simple objects in $\Rep A$ by using the free functor $F:\C\to \Rep A$ defined by $F(X)=(A\ox X,\mu\ox\id_X)$. Depict the module action $\mu\ox\id_X$ as
\begin{equation}\label{eq:FXaction}
	\mu\ox\id_X \,= \, \FXaction.
\end{equation}
A direct computation implies the quantum dimension $\dima FX=\dimc X$. With the tensor structure $\eta_{X,Y}:FX\ox_A FY\rarr F(X\ox Y)$ given by
\begin{equation}\label{eq:naturalIsoTensorF}
	\eta_{X,Y} =	\naturalIsoTensorF,
\end{equation}
$F$ is a tensor functor.

Each $(V,\rho)\in \Rep A$ is a direct summand of $FV$. To see this, we construct $e:FV\rarr (V,\rho),f:(V,\rho)\rarr FV$ by
\begin{equation}\label{eq:FVdirectSummand}
	\FVdirectSummand
\end{equation}
satisfying that $e\circ f=\id_V$ and that $f\circ e$ being an idempotent. Hence, all objects in $\Rep A$ are subobjects of $FX$ for some $X\in\C$. We denote such subobjects by $\obj{FX} {\overline p}$, where $\overline p:FX\rarr FX$ is an idempotent arising from $p: A\ox X\rarr X$, defined as
\begin{equation}\label{eq:FXidempotent}
	\overline p = \quad\FXidempotent.
\end{equation}
The idempotent condition $\overline p \,\overline p = \overline p$ is equivalent to
\begin{equation}\label{eq:FXidempotentAA}
	\FXidempotentAA.
\end{equation}
Every idempotent $FX\to FX$ must be in the form $\overline p$ for some $p:A\ox X\rarr X$, according to the property \eqref{eq:AmorphismExpectation}.  Every $(V,\rho)\in \Rep A$ can be identified with $\obj{FV}{\frac{1}{\dimc A}\,\overline\rho}$ in the canonical way \eqref{eq:FVdirectSummand}. 

Simple objects are now written as $\obj{FX}{\overline p}$ where $\overline p:FX \rarr FX$ is a minimal idempotent (namely, $\overline p$ can not be $\overline {p_1}+\overline{p_2}$ for $\overline {p_1}\neq 0$ and $\overline{p_2}\neq 0$). Particularly, the unit $(A,\mu)$ is identified with $F\unit|_{\overline{\epsilon}}$. 


A useful property is
\begin{equation}\label{eq:HomFXW}
	\Hom_A(FX,W)\cong \Hom_\scC(X,W),
\end{equation}
for every $W\in\Rep A$ and $X\in\scC$. Indeed, every morphism $f\in \Hom_A(FX,W)$ can be written uniquely in terms of a morphism in $\Hom_\scC(X,W)$ by
\begin{equation}\label{eq:morphismFXW}
	\morphismFXW
\end{equation}

Particularly, for $X,Y\in \mathcal{C}$, we have
\begin{equation}\label{eq:HomFXFY}
	\Hom_A(FX,FY)\cong \Hom_\scC(X,A\ox Y).
\end{equation}
The image of $F:\mathcal{C}\to \Rep A$ is equivalent to the pre-quotient category $\mathcal{C}/A$, whose objects are those of $\mathcal{C}$, and whose morphisms are
\begin{equation}\label{eq:morphismQuotient}
	\Hom_{\scC/A}(X,Y)=\Hom_\scC(X,A\ox Y).
\end{equation}
When the objects are equipped with idempotents, they form the quotient category.

\section{Reconstruction theorem: finite group $G=\Aut A$ from the absorbing Frobenius algebra}
\label{sec:reconstruct}

When $A$ is absorbing (to be defined below), we can reconstruct\cite{Deligne1991,Muger2004} a finite group $G$ such that $A$ is equivalent to the regular representation of $G$. 

The $A$ is \textit{absorbing} if
\begin{equation}\label{eq:absorbingConditionFrobenius}
	(A,\mu)\ox (A,\mu)=(A,\mu)^{\oplus N}
\end{equation}
for some positive integer $N$. Let $\scS$ be the minimum full subcategory of $\mathcal{C}$ containing $A$. In the means time, denote $L_{\scS}\subset L_{\mathcal{C}}$ as the subset of (isomorphism class of) simple objects in $\scS$.


According Perron-Frobenius theorem, $A$ must have the form
\begin{equation}\label{eq:regualrRepresentationFA}
	A=\bigoplus_{X_j\in L_{\scS}} {X_j}^{\oplus \PFdimc(X_j)},
\end{equation}
where $\PFdimc X_j\in \Z_+$ is the Penrron-Frobenius dimension. See Appendix \ref{sec:dimensionProof} for the proof.

The $\scS$ is a symmetric tensor category since, $\forall X_j\in\scS$, $\theta_{X_j}=\id_{X_j}$ holds. In [Muger], they have constructed a fibre functor using $A$, which is used to reconstruct a finite group $G$ such that $\scS$ is equivalent to the Tannakian representation category $\Rep_G$. We briefly review this reconstruction below.

The fibre functor $E:\scS\rarr \mathrm{Vect}$ is constructed as follows. For every $X\in\scS$, we define $EX=\Hom_\C(A,X)$, and $E(f)$ for $f:X\rarr Y$ by $E(f)\phi=f\circ \phi$ for $\phi\in EX$. The tensor structure $\varphi_{X,Y}:EX\ox EY\rarr E(X\ox Y)$ is given by
\begin{equation}\label{eq:naturalEXEYEXY}
	\naturalEXEYEXY.
\end{equation}
Note that for $\varphi_{X,Y}$ to be an isomorphism, we need the absorbing condition \eqref{eq:absorbingConditionFrobenius}.

Define
\begin{equation}\label{eq:AutAGroup}
	\Aut(A)=\{g\in \mathrm{End}_{\mathcal{C}}(A) |\quad \autAcond\quad\}
\end{equation}

There is a group isomorphism $\Aut(A)\rarr \mathrm{Nat}^{\ox}E$ given by
\begin{equation}\label{eq:homormorphismAutAtoNatE}
	g\mapsto g_X, g_X(\phi)=\phi\circ \ginv{g},\phi\in EX.
\end{equation}

See \cite{Muger2004} for more details about this reconstruction. Below we present the result in a slightly different form.

\begin{proposition}\label{prop:absorbingA}
	Let $A$ be a strongly separable commutative Frobenius algebra in $\mathcal{C}$. The $G=\Aut(A)$ is a finite group such that $\scS$ is equivalent to the Tannakian representation category $\Rep_G$, if
	\begin{equation}\label{eq:AA=nA}
		A\ox A=A^{\oplus N},
	\end{equation}
	for $N\in \mathbb{Z}_+$, and
	\begin{equation}\label{eq:groupCompleteCondition}
		\groupCompleteCondition,
	\end{equation}
	\begin{equation}\label{eq:groupOrthogonalCondtion}
		\groupOrthogonalCondition \,= |G|\,\delta_{g,1},
	\end{equation}
	Here, Eq. \eqref{eq:groupOrthogonalCondtion} is a consequence of Eq. \eqref{eq:groupCompleteCondition}.
\end{proposition}
\begin{proof}
	Eqs. \eqref{eq:groupCompleteCondition} and \eqref{eq:groupOrthogonalCondtion} is equivalent to the absorbing property \eqref{eq:absorbingConditionFrobenius}.
	
	(1). Eqs. \eqref{eq:groupCompleteCondition} and \eqref{eq:groupOrthogonalCondtion} imply the absorbing property \eqref{eq:absorbingConditionFrobenius}. Define $e_g:(A,\mu)\ox (A,\mu)\to(A,\mu),f_g:(A,\mu)\to (A,\mu)\ox (A,\mu)$ by
	\begin{equation}\label{eq:decomposeFAidempotent}
		\decomposeFAidempotent.
	\end{equation}
	Eqs. \eqref{eq:groupCompleteCondition} and \eqref{eq:groupOrthogonalCondtion} imply $\sum_g f_g\circ e_g=\id_{A\ox A}$ and $e_g\circ f_g=\id_{A}$, depicted by
	\begin{equation}\label{eq:FAorthonormal}
		\FAorthonormalCondition
	\end{equation}
	\begin{equation}\label{eq:FAorthonormalAA}
		\FAorthonormalConditionAA
	\end{equation}
	Eq. \eqref{eq:FAorthonormalAA} is a consequence of Eq. \eqref{eq:FAorthonormal}. Denote the left hand side of Eq. \eqref{eq:FAorthonormal} by $P$, then $P\circ P=\id_{A\ox A}$ implies Eq. \eqref{eq:FAorthonormalAA}.

	(2). The absorbing property \eqref{eq:absorbingConditionFrobenius} imply Eqs. \eqref{eq:groupCompleteCondition} and \eqref{eq:groupOrthogonalCondtion}. By the absorbing property \eqref{eq:absorbingConditionFrobenius}, there exists $\alpha_n,\alpha_{n}^*\in\End_{\mathcal{C}}(A)$, $1\leq n\leq N=\PFdimc A$, such that
	\begin{equation}\label{eq:absorbingDecomposition}
		\absorbingDecomposition.
	\end{equation}
	
	Because $|G|=\sum_{X_j\in L_{\scS}}(\dim EX_j)^2=\sum_{X_j\in L_{\scS}}(\PFdimc X_j)^2=\PFdimc A=\dim \End_{\mathcal{C}}(A)$, any of $\{g\in G\}$, $\{\alpha_n\}$ and $\{\alpha^*_n\}$ is a basis of $\End_{\mathcal{C}}(A)$. Hence $\alpha_n$ and $\alpha_n^*$ have linear expansions $\alpha_n =\sum_{g\in G}\alpha_n(g)g$ and $\alpha^*_n=\sum_{h\in G}\alpha^*_n(h)h$ for $\alpha_n(g),\alpha^*_n(h)\in\mathbb{C}$. Eqs. \eqref{eq:absorbingDecomposition} imply
	\begin{equation}\label{eq:proofAbsorbingProposition}
		\id_A=\sum_n \alpha_n\circ \alpha^*_n=\sum_{g,h\in G}\left(\sum_n \alpha_n(g)\alpha^*_n(h)\right)gh,
	\end{equation}
	which holds iff $\sum_n \alpha_n(g)\alpha^*_n(h)=\delta_{gh,1}/|G|$, and Eq. \eqref{eq:proofAbsorbingProposition} leads to Eq. \eqref{eq:FAorthonormal}.

\end{proof}

A useful consequence is that any $\mathcal{C}$-morphism $f:A\rarr A\ox A$ can be decomposed as
\begin{equation}\label{eq:decomposeMorphismAAA}
	\decomposeMorphismAAA
\end{equation}

In the above reconstruction, $A$ is interpreted as a left regular representation space of $G$, and the properties \eqref{eq:FAorthonormal} and \eqref{eq:groupOrthogonalCondtion} are the usual orthonormal conditions in the representation theory.


\subsection{Dimension Proof}
\label{sec:dimensionProof}

Suppose the decomposition of $A$ is
\begin{equation}\label{eq:AinTermsOfW}
	A=\bigoplus_{X_j\in L_{\scS}} n_j X_j,
\end{equation}
where multiplicity $n_j\in \Z_{+}$, and $n_{\unit}=1$. 

The decomposition $A\ox A=A^{\oplus N}$ implies 
\begin{equation}
	\sum_{ij\in L_{\scS}}N_{ij}^k n_in_j=Nn_k,
\end{equation}
where $N_{ij}^k=\dim \Hom_\C(X_i\ox X_j,X_k)$. Define a $|L_\scS|\times|L_\scS|$ matrix $R$ whose matrix element is $R_{kj}=\sum_{i}N_{ij}^k n_i$. This $R$ is strictly positive matrix (otherwise $N^k_{ij}=0$ for all $i$). By Perron-Frobenius theorem, there is a unique eigenvector with strictly positive entries $n_k$. Using the normalization $n_{\unit}=1$, the only solution is $n_j=\PFdimc X_j$, where $\mathrm{FPdim}Z_j$ is the Frobenius-Perron dimension. This dimension is defined by the maximal non-negative eigenvalue of the matrix $N_i$ whose matrix elements are $[N_i]_{jk}=N_{ij}^k$.

By everything mentioned above, we conclude that
\begin{equation}
	A=\bigoplus_{X_j\in L_{\scS}}  {X_j}^{\oplus \PFdimc (X_j)}.
\end{equation}

\section{Group Action and Equivariantization}

\subsection{Grading of $\Rep A$ by $G=\Aut A$}
\label{sec:RepAGradedByG}

When $A$ is absorbing, a finite group $G=\Aut(A)$ is reconstructed. We review the $G$-grading of $\Rep A$.

Let $G$ be a finite group. A $G$-graded fusion category $\mathcal{D}$ is a fusion category endowed with pairwise disjoint subcategories $\{\mathcal{D}_g\}_{g\in G}$ such that 

\begin{enumerate}
	\item[(a)] every $X\in \mathcal{D}$ splits as $X=\oplus_g X_g$ where $X_g\in \mathcal{D}_g$;
	\item[(b)] if $X\in\mathcal{D}_g,Y\in\mathcal{D}_h$, then $X\ox Y\in \mathcal{D}_{gh}$;
	\item[(c)]  if $X\in\mathcal{D}_g,Y\in\C_h$ with $g\neq h$, then $\Hom_{\mathcal{D}}(X,Y)=0$;
	\item[(d)] $\unit_{\mathcal{D}}\in\mathcal{D}_1$. 
\end{enumerate}

For $X\in \mathcal{D}_g$, we say $X$ has the grading $\grad X=g$. Note that condition (b) implies $\grad(X^*)=\grad(X)^{-1}$, because $X\ox X^*=\unit_\mathcal{D}\oplus\dots$ and $\grad ({\unit_\mathcal{D}})=1$.

Let $\mathcal{D}=\Rep A$ and $G=\Aut(A)$. Before constructing the $G$-grading, we introduce a useful idempotent $\pi_g:FX\to FX$, $X\in\mathcal{C}$:
\begin{equation}\label{eq:gradingMiddleIdempotent}
	\pi_g=\,\gradingMiddleIdempotent.
\end{equation}
It is possible that: (1) $\pi_g=0$, or (2) $\pi_g$ is not minimal.
\begin{proposition}
	\label{prop:pi_g}
	If $p:FX\to FX$ is a minimal idempotent, then there exists a unique $g\in G$ such that $p\pi_g=p=\pi_gp$.
\end{proposition}
\begin{proof}
	Obviously, $\pi_g$ is central: $\pi_gp=p\pi_g$. Also use $\sum_g\pi_g=\id_{A\ox X}$ and $\pi_g\pi_h=0$ if $g\neq h$.
\end{proof}

For every simple $X\in \Rep A$, we define the grading by
\begin{equation}\label{eq:gradingXinRho}
	\grad(X)=\frac{1}{\dimc X}\gradingXinRho.
\end{equation}
The definition extends to a generic object of $\Rep A$ in the obvious way.

To prove $\grad(X)\in G$ we compute
\begin{equation}\label{eq:gradingProofAA}
	\gradingXProofAA,
\end{equation}
where we use Eq. \eqref{eq:FAorthonormal}. The right hand side contains an idempotent $\pi_g$. According to Proposition \ref{prop:pi_g} there exists a unique group element, denoted by $\grad(X)\in G$, such that
\begin{equation}\label{eq:gradingXProofAB}
	\gradingXProofAB,
\end{equation}
which leads to
\begin{equation}\label{eq:gradingXFinal}
	\frac{1}{\dimc X}\gradingXinRho =\,	\gradingXFinal.
\end{equation}

From the above proof, we obtain two useful identities:
\begin{equation}\label{eq:moreGeneralGradingIdentity}
	\moreGeneralGradingIdentity.
\end{equation}

Define $\Rep^0A$ as the subcategory of $\Rep A$ with objects $V\in \Rep A$ satisfying
\begin{equation}\label{eq:commutingWithA}
	\defRepOA
\end{equation}
According to Eqs. \eqref{eq:moreGeneralGradingIdentity}, Eq. \eqref{eq:commutingWithA} is equivalent to
\begin{equation}\label{eq:gradingRepUnitA}
	\grad(V)=1.
\end{equation}
Hence, $\Rep^0A$ is the subcategory with grading $g=1$.

Hence the $\Rep A$ is graded by a finite group $G=\Aut(A)$, with the unit component being $\Rep^0 A$.

The property
\begin{equation}\label{eq:gradingZaoxZ}
	\grad(X\aox Y)=\grad(X)\circ\grad(Y)
\end{equation}
can be verified directly
\begin{align}\label{eq:gradingPropertyProofAB}
	&\,\grad(X\aox Y)
	\nonumber\\
	=\,&\gradingPropertyProofsAB\,
	\nonumber\\
	=\,&\gradingPropertyProofsAC\,
	\nonumber\\
	=\,&\grad(X)\circ\grad(Y).
\end{align}

For simple $X,Y\in\Rep A$, suppose $X\aox Y=Z_1\oplus Z_2\oplus\dots$, then
\begin{equation}\label{eq:XYZ1Z2}
	\dimc X\dimc Y\grad(X\aox Y)=\dimc Z_1\grad Z_1 + \dimc Z_2 \grad Z_2.
\end{equation}
Because $\{g\}$ is linear independent basis of $\End_{\mathcal{C}}(A)$, we have
\begin{equation}\label{eq:gradZ1gradZ2}
	\grad Z_1=\grad Z_2=\dots=\grad(X)\circ\grad(Y).
\end{equation}
An immediate consequence is
\begin{equation}\label{eq:gradingZinv}
	\grad(X^*)=\grad(X)^{-1}.
\end{equation}

\subsection{Group action and Equivariantization}
\label{sec:IsomorphismPsi}

Using the $G=\Aut(A)$ reconstructed from $A$, we will construct a group action and the equivariantization category $(\Rep A)^G$.

We briefly recall some notation on tensor categories. A tensor functor between two tensor categories $\mathcal{C}$ and $\mathcal{D}$ is $(F,\varphi,\varphi_0)$ where $F:\mathcal{C}\rarr\mathcal{D}$ is a functor, $\varphi:F\circ \ox_\scC\isoto \ox_\scD\circ(F\times F)$ a natural isomorphism, and $\varphi_0: F(\unit_\scC)\isoto\unit_\scD$ is an isomorphism. We call $\varphi$ the tensor structure on $F$ and $\varphi_0$ the unit-preserving structure on $F$. For a finite group $G$, let $\underline{G}$ be the tensor category whose objects are elements of $G$, morphisms are the identities, and whose tensor product is given by the multiplication in $G$. 

A group action on $\scC$ is $(T,\gamma,\iota,\mu,\nu)$ where $T$ is a tensor functor $T:\underline{G}\rarr \Aut_{\ox}\mathcal{D}; g\mapsto T_g$. The category $\underline{G}$ has objects $g\in G$, and morphisms $\Hom(g,h)=\emptyset$ if $g\neq h$ and $\Hom(g,g)=\mathbb{C}$, with $\ox$ defined by $g\ox h=gh$. $\Aut_{\ox}\mathcal{D}$ is the category of auto tensor functors $\mathcal{D}\to \mathcal{D}$. The tensor structure and the unit-preserving structure on $T$ are denoted by
\begin{equation}\label{eq:tensorStructureVarphi}
	\gamma_{g,h}: T_{gh}\isoto T_g\circ T_h
\end{equation}
\begin{equation}\label{eq:unitPreservingStructureVarphi}
	\iota: T_{e}\isoto \id_{\scC}
\end{equation}
for any $g,h\in G$. And $\{\mu_g\}$ is a family of tensor structures on $T_g$, $\{\nu_g\}$ a family of unit-preserving structures on $T_g$.

A $G$-\textit{equivariant} object in $\mathcal{D}$ is a pair $(X,\{u_g\}_{g\in G})$ where $X\in\mathcal{D}$ and $u_g:T_g X\to X$ is a family of isomorphisms called \textit{equivariant structure} on $X$ such that
\begin{equation}\label{eq:EquivariantStructureSquare}
	\EquivariantStructureSquare
\end{equation}
for all $g,h\in G$.

The fusion category formed by the $G$-equivariant objects is called the \textit{equivariantization} of $\mathcal{D}$, denoted by $\mathcal{D}^G$. The morphisms are those of $\mathcal{D}$ that commute with $u_g$.

Let $L_{\Rep A}=\{(V,\rho_V)\}$ be the set of representatives of simple objects in $\mathrm{Rep}A$. 

For every $(V,\rho_V)\in \Rep A$, we define $T_g(V,\rho_V)=(V,\rho_V^{\ginv{g}})\in \Rep A$, where
\begin{equation}\label{eq:transformationRhogApp}
	\uniqueMorphismVVprimeAC.
\end{equation}
If $(V,\rho_V)$ is simple, $T_g(V,\rho_V)$ is simple too.

\begin{proposition}
For every $(V,\rho_V)\in L_{\Rep A}$ and $g\in G$, there exits an object $(V,\rho'_V)\in L_{\Rep A}$, and a unique $\mathcal{C}$-isomorphism $u_g:V\to V$, such that
\begin{equation}\label{eq:uniqueMorphismVVprimeABapp}
	\uniqueMorphismVVprimeAB.
\end{equation}
\end{proposition}
\begin{proof}
	Proof: The $T_g(V,\rho_V)$ is simple, hence there exits an $A$-isomorphism $u_g:(V,\rho_V^{\ginv{g}})\to (V,\rho'_V)$ for some object $(V,\rho'_V)\in L_{\Rep A}$. The $u_g$ satisfies Eq. \eqref{eq:uniqueMorphismVVprimeABapp}. By Schur's lemma, $u_g$ is unique up to a constant.
\end{proof}

The $u_g$ is called the \textit{equivariant structure}, which will be justified soon.

For simple $(V,\rho_V)\in \Rep A$,
\begin{equation}\label{eq:uguhugh}
	\uguhugh,
\end{equation}
according to the uniqueness of $u_g$.

In the following, we suppress $(V,\rho_V)$ to $V$ for simplicity.

For simple $U,V,W\in\Rep A$, let $u_g:U\to U', V\to V'$ and $W\to W'$ be defined as above. We have an isomorphism $\mu_g(U;V,W):\Hom_A(U',V'\aox W')\to \Hom_A(U,V\aox W)$ defined by
\begin{equation}\label{eq:mugUVW}
	\mugUVW.
\end{equation}

Since $\Rep A$ is semisimple, $V\aox W$ and $V'\aox W'$ can be decomposed as
\begin{equation}\label{eq:isotropyVWVW}
	\isotropyVWVW.
\end{equation}
We use the same $\alpha$ to label the basis because $\Hom_A(U',V'\aox W') \cong \Hom_A(U,V\aox W)$ according Eq. \eqref{eq:mugUVW}. We require $\mu_g(U;V,W)(\alpha')=\eta\alpha$ for some nonzero complex number $\eta$.

Define $\mu_g(V,W):V\aox W\to V\aox W$ by
\begin{equation}\label{eq:mugUV}
	\mugUAA.
\end{equation}

We define $u_g: V\aox W\to V'\aox W'$ by
\begin{equation}\label{eq:ugVoxW}
	\ugVoxW.
\end{equation}

Given Eq. \eqref{eq:ugVoxW}, we can extend the definition $u_g$ to generic objects of $\Rep A$. Using $u_g$, we define the group action and the equivariantization below.

For every $g\in G$, we define the tensor functor $T_g:\Rep A\to \Rep A$ as follows.
Let $T^g(V,\rho_V)=(V,\rho_V^{\ginv{g}})$ as defined in Eq. \eqref{eq:transformationRhogApp}, and $T_g(f)=f$ for $f\in \Hom_A(V,V)$. Define the tensor structure on $T_g$ by $\mu_g(V,W):T_g(V\aox W)\to T_g(V)\aox T_g(W)$ as given in Eq. \eqref{eq:mugUV}.

Then we define the group action $T:\underline{G}\to \Aut_\ox(\Rep A)$, by defining the tensor structure on $T$ by
\begin{equation}\label{eq:tensorStructureOnT}
	\beta_V(g,h)^{-1}\id_V: T_{gh}V\to T_gT_hV,
\end{equation}
where $\beta_V(g,h)$ is defined in Eq. \eqref{eq:uguhugh}.

Now we claim $u_g:T_gV\to V$ is an equivariant structure. The equivariantization cateogry $(\Rep A)^G$ has objects $(V\in \Rep A,u_g:V\to V)$. 


The grading under the group action is given by
\begin{equation}\label{eq:conjugatePhigZ}
	\grad(T_gX)=g\circ \grad(X) \circ \ginv{g}
\end{equation}
which is verified by
\begin{equation}\label{eq:gradingPropertyProofAA}
	\grad(T_gX)\,=\,\gradingPropertyProofsAA\,=\,g\circ \grad(X) \circ \ginv{g}
\end{equation}

\subsubsection{Alternative Group action}
\label{sec:AlternativeGroupAction}

Using the group action $T$ has the advantage of nice properties in Eqs. \eqref{eq:AutAGroup}\eqref{eq:groupCompleteCondition}\eqref{eq:groupOrthogonalCondtion}. But $T_g(V,\rho_V)\neq L_{\Rep A}$ for any $(V,\rho_V)$, hence $T$ is not convenient for a computations in a tensor description. The equivariant structure $u_g$ maps back $T_g(V,\rho_V)$ to some $(V,\rho'_V)\neq L_{\Rep A}$. Alternatively, one can define $T'$ by $T'_g(V)=V'$ for every $V\in\Rep A$ with $u_g:V\to V'$, and $T_g(f)=u_g\circ f\circ u_{g}^{-1}$ for every $A$-morphism $f$. We will not use this construction in this paper.

\subsection{Concrete construction of the equivariant structure}

In the following, we construct an explicit form of the equivariant structure $u_g$. For every $g$, we pick up a representative simple object in $L_{\Rep A}$ with the grading $g$, denoted by $V_g$.

For every $W\in \Rep A$ and every $g$, we have a projector
\begin{equation}\label{eq:projectionVWV}
	P^g=\frac{1}{\dim_{\mathcal{C}}V_g} 
	\,\, 
	\projectionVWV
\end{equation}
Since $\rep A$ is semisimple, there exists $W'\in \Rep A$ and the morphisms
\begin{align}
	\projectionIsotropyAA \in \Hom_A(V_g\aox W\aox V_g^*,W'),
	\\
	\projectionIsotropyAB \in \Hom_A(W',V_g\aox W \aox V_g^*),
\end{align}
such that
\begin{equation}
	\projectionIsotropyAC,
\end{equation}
\begin{equation}
	\projectionIsotropyAD.	
\end{equation}

In the following, we will construct an isomorphism $u_g:T_g W\to W'$. Since $W'$ is simple, we choose $W'\in L_{\Rep A}$ from now on.

Define $u_g$ to be
\begin{equation}\label{eq:ugProjection}
	u_g = \frac{1}{\sqrt{\dim_{\mathcal{C}} V_g}} \defineUgProjection.
\end{equation}
We check it satisfies the defining property \eqref{eq:transformationRhogApp}:
\begin{align}
	\defineUgProjectionAA = & \defineUgProjectionAB
	= \defineUgProjectionAC
	\nonumber \\
	= & \defineUgProjectionAC 
	=  \defineUgProjectionAD
	\nonumber \\
	= & \defineUgProjectionAE
	=  \defineUgProjectionAF.
\end{align}
The $u_g^{-1}$ is 
\begin{equation}\label{eq:ugInverse}
	u_g^{-1} =  \frac{1}{\sqrt{\dim_{\mathcal{C}} V_g}} \defineUgInverse.
\end{equation}

\subsection{Orbits in $\Rep A$ under the group action}\label{sec:orbitsUnderGroupAction}

\begin{proposition} \label{prop:FVdecomposition}
	Simple objects of $\Rep A$ has properties:
	\begin{enumerate}[label=(\roman*)]
		\item For simple $(V,\rho_V)\in\Rep A$, $FV=\displaystyle\bigoplus_g T_g(V,\rho_V)$.
		\item For simple $(V,\rho_V),(V,\rho'_V)\in\Rep A$, there exists a non-empty subgroup $K_{(V,\rho_V)}\subseteq G$, such that $T_g (V,\rho_V)=(V,\rho'_V)$ for $g\in K_{(V,\rho_V)}$, and $|K_{(V,\rho_V)}|=\dim \Hom_{\mathcal{C}}(V,V)$.
		\item For simple $(V,\rho_V),(W,\rho_W)\in\Rep A$, either $V = W$ or $\Hom_{\mathcal{C}}(V,W)=0$.
	\end{enumerate}
\end{proposition}
\begin{proof}
	\begin{enumerate}[label=(\roman*)]
		\item Use the decomposition
		\begin{equation}\label{eq:FVdecomposition}
			\FVdecomposition.
		\end{equation}
		\item Using Eq. \eqref{eq:HomFXW}, we have
		\begin{equation}\label{eq:rho2rhoPrime}
			\Hom_A\left(\bigoplus_g T_g(V,\rho_V),(V,\rho'_V)\right)=\Hom_A\left(FV,(V,\rho'_V)\right)=\Hom_{\mathcal{C}}(V,V)
		\end{equation}
		\item Suppose $\Hom_{\mathcal{C}}(V,W)\neq 0$. Then 
		\begin{align}
			&\Hom_A(FV,(W,\rho_W))=\Hom_{\mathcal{C}}(V,W)\neq 0\nonumber\\
			\Rightarrow\quad &T_g(V,\rho_V)=(W,\rho_W)\text{ for some }g\in G
			\nonumber\\
			\Rightarrow \quad & V=W.
		\end{align}
		
	\end{enumerate}
\end{proof}

\section{Twisted quantum double $D^\omega G$ and its representations}
\label{sec:tqdAndreps}

\subsection{Twisted quantum double $D^\omega G$}

In this section, we briefly review the definition of the twisted Drinfeld's twisted double of a finite group\cite{Dijkgraaf1991} as a quasi-triangular quasi-Hopf algebra. Let $D^{\alpha } G$ be a finite-dimensional vector space with a basis $\{\qdA^g\qdB^x\}_{(g,x)\in G\times G}$. Define a product on $D^{\alpha } G$ by
\begin{equation}
	\label{eq:tqdMultiplication}
	( \qdA^g\qdB^x)( \qdA^h\qdB^y) :=\delta _{x,hy\bar{h}} \beta_y(g,h) \qdA^{gh}\qdB^y.
\end{equation}
We also use symbols $\qdB^h=\qdA^1\qdB^h$, $\qdA^g=\sum_x \qdA^g\qdB^x$, and $\qdB^h \qdA^g = \qdA^g \qdB^{\ginv{g}hg}$ for convenience.
The unit is
\begin{equation}
	1=\sum _{x\in G} \qdA^1 \qdB^x.
\end{equation}
Define the coproduct $\Delta:D^{\alpha } G\rightarrow D^{\alpha } G\otimes D^{\alpha } G$ and the counit $\varepsilon:D^{\alpha } G\rightarrow \mathbb{C}$ by
\begin{equation}\label{eq:coproductTQD}
	\Delta ( \qdA^g\qdB^x) :=\sum _{a,b\in G:ab=x} \mu _{g} (a,b) (\qdA^g\qdB^a) \otimes (\qdA^g\qdB^b)
\end{equation}
and
\begin{equation}
	\varepsilon ( \qdA^g\qdB^x) :=\delta _{x,1}.
\end{equation}
Besides, set the Drinfeld associator by
\begin{equation}
	\Phi :=\sum _{x,y,z\in G} \alpha (x,y,z)^{-1} (\qdA^1\qdB^x)\otimes (\qdA^1\qdB^y)\otimes (\qdA^1\qdB^z),
\end{equation}	
and the $R$ matrix\footnote{we choose $R$ matrix to be compatible with formula in Section \ref{sec:chargeCondensationTQD}} by
\begin{equation}
	\label{eq:Rmatrix}
	R:=\sum _{x,y\in G} \frac{1}{\beta_{y}(x,x^{-1})}((\qdA^{x^{-1}})\qdB^y)\otimes (\qdA^1\qdB^x).
\end{equation}
Finally, define the antipode by a linear map $S:D^{\alpha } G\rightarrow D^{\alpha } G$
\begin{equation}\label{eq:antipodeTQD}
	S( \qdA^g\qdB^x) :=\frac{1}{\beta_{\bar{x}}\left({g}, \bar g\right) \mu _{g}\left( x,\bar{x}\right)} \qdA^{\bar{g}}\qdB^{{g}\bar{x}\bar{g}},
\end{equation}
where $\beta_x(g,h)$ and $\mu _{g} (x,y)$ are expressed in terms of $\alpha $ by
\begin{equation}
	\beta_x(g,h):=\frac{\alpha (g,h,x)\alpha \left( ghx\bar{h} \bar{g} ,g,h\right)}{\alpha \left( g,hx\bar{h} ,h\right)}
\end{equation}
and
\begin{equation}
	\label{eq:tqdMu}
	\mu _{g} (x,y):=\frac{\alpha \left( gx\bar{g} ,gy\bar{g} ,g\right) \alpha (g,x,y)}{\alpha \left( gx\bar{g} ,g,y\right)}
\end{equation}
for all $g,h,x,y\in G$.

\subsection{Representations}
\label{sec:repTQD}

In this subsection, we review the simple objects in $\mathrm{Rep}(D^\omega G)$. For every conjugacy class $C$ of $G$, we choose a representative member $h_1\in C$. We denote by $Z_C$ the centralizer $\{z\in G|z h_1= h_1z\}$. We denote by $(\mu,V_\mu)$ the irreducible projective representations of $Z_C$ with $\beta_{h_1}\in H^2(Z_C,U(1))$:
\begin{equation}\label{eq:projectiveRepresentationZC}
	\mu(z_1)\mu(z_2)=\beta_{h_1}(z_1,z_2)\mu(z_1z_2),
\end{equation}
where $\beta$ is defined by
\label{eq:betaFormula}
\begin{equation}
	\beta_{h_1}(z_1,z_2):=\frac{\omega (z_1,z_2,{h_1})\omega \left( z_1z_2{h_1}{z_2}^{-1}{z_1}^{-1} ,z_1,z_2\right)}{\omega \left( z_1,z_2{h_1}{z_2}^{-1} ,z_2\right)}.
\end{equation}

For every element ${h_i}\in {C}$, we introduce a unique representative $x_i\in G$ such that $x_i({h_1})\inv{x_i}={h_i}$. For convenience, we set $x_1=1$.

The simple objects in $\mathrm{Rep}(D^\omega G)$ are labeled by $(C,\mu)$. For every such pair $(C,\mu)$, we define an irreducible representation $(\pi^C_{\mu},V^C_{\mu})$ of $D^{\omega}(G)$, where $V^C_{\mu}=\mathbb{C}[C]\otimes V_{\mu}$ is the representation space. We denote by $\qdA^g\qdB^h$ the generator elements of $D^\omega G$, and define the representation matrix by
\begin{equation}
	\label{eq:repTQD}
	\pi^C_{\mu}(\qdA^g\qdB^h)\ket{h_i,v}=\delta_{h,{h_i}}\frac{\beta_{\inv{x_i}hx_i}(g,x_i)}{\beta_{\inv{x_i}hx_i}(x_k,\tilde g)}\ket{h_k,\mu(\inv{x_k}gx_i)v},
\end{equation}
where $v\in V_{\mu}$ and $x_k$ is determined by $g({h_i})\inv{g}={h_k}$, and $\tilde g=\inv{x_k}gx_i$.

\bibliographystyle{apsrev4-1}
\bibliography{StringNet}
\end{document}